\documentclass{article}
\usepackage{amsmath,amsbsy,amsfonts,amssymb,amsthm,amscd,amsxtra,amstext,amsopn}
\usepackage{mathtools,emptypage,dsfont,yfonts,latexsym,mathrsfs,bbold}
\usepackage[left=0.6in,right=1.1in]{geometry}
\usepackage[fontsize=11pt]{scrextend}
\usepackage{xcolor}
\usepackage{comment}
\usepackage{marginnote}

\newcommand \ul[1]{\underline{#1}}
\renewcommand\v[1]{\vec{#1}}

\newcommand \Id{\mathrm{Id}}

\newcommand*\R{ \mathds{R} }

\newcommand*\N{ \mathds{N} }

\newcommand*\cL{ \mathcal{L} }

\newcommand\uv{{\underline v}}
\newcommand\uu{{\underline u}}

\newcommand\uxi{{\underline \xi}}
\newcommand\ueta{{\underline \eta}}

\newcommand\vv{{\vec v}}
\newcommand\vu{{\vec u}}
\newcommand\vw{{\vec w}}
\newcommand\vxi{{\vec \xi}}
\newcommand\veta{{\vec \eta}}

\newtheorem{mylemma}{Lemma}

\newtheorem{Prop}{Proposition}

\newtheorem{Remark}{Remark}

\newtheorem{mythm}{Theorem}

\newtheorem{myclaim}{Claim}

\newtheorem{corollary}[mythm]{Corollary}

\DeclarePairedDelimiter{\scal}{\langle}{\rangle}

\title{A Kac system interacting with two heat reservoirs: the shearing case}
\author{{\bf Federico Bonetto$^{1,*}$ and Matthew Powell$^1$}\\
\small $^1$ School of Mathematics, Georgia Institute of Technology,Atlanta, GA 30332, USA\\
\small $^*$ corresponding author: \tt bonetto@math.gatech.edu}
\date{\today}

\begin{document}
\maketitle

\begin{abstract}
We study a system formed by $M$ particles moving in 3 dimensions and interacting with two heat reservoirs, each with
$N\gg M$ particles. The system and the reservoirs interact via random collisions and thus evolve via a Kac-type master
equation. The initial state of the reservoirs is given by two non-centered Maxwellian distributions; they have
temperature $T_+$ and $T_-$ and have average velocity $\vec p_+$ and $\vec p_-$, respectively. We prove
that, for times shorter than $\sqrt{N}/M$, the interaction with the two reservoirs is well-approximated by the
interaction with two shearing {\it dynamic} Maxwellian thermostats (i.e. heat reservoirs with $N=\infty$). As a
byproduct of our analysis, we obtain a uniform in time approximation when $T_+=T_-$ and $\vec p_+=\vec p_-$.
\end{abstract}

\begin{center} {\bf Acknowledgement}\\
\small We are indebted to Omar Hurtado, Michael Loss and Tobias Ried for many enlightening discussions.
\end{center}


\section{Introduction, Model and Results}

The Kac model describes particles interacting via random binary collisions meant to give a simplified picture of what
happens in a real gas. The model was introduced in 1953 by Mark Kac\cite{Kac}; the hope was to use the model to prove
the existence of solutions for the Boltzmann equation. Although this has not been realized, the model has proven very
useful in several different respects. In kinetic theory it helped reveal how systems approach equilibrium and the 
relation between the different norms one can use to measure such an approach, see \cite{CCL1,CCL2} for convergence in 
$L^2$, \cite{PS18} for convergence in $L^1$ and \cite{Hagop} for the GTW metric $d_2$. It also helped clarify the 
conditions for the emergence of macroscopic behaviors like propagation of chaos and the validity of the Boltzmann-Kac 
equation, see e.g. \cite{MischlerMouhot1,MischlerMouhot2,Sznitman84}. Finally, the ideas behind the Kac 
evolution are partially present in the DSMC scheme for numerical simulations, see \cite{B94}. 

In a series of papers (see \cite{BLV,BLTV,BGLR,BLP25}) we have tried to extend the reach of the Kac model by adding
interactions both with reservoirs (still modeled via random collisions) and Maxwellian thermostats (seen as infinite
reservoirs), and by looking at particles moving in three spacial dimensions, instead of one as in the original model. In
particular, in \cite{BLP25} we looked at a three dimensional system with $M$ particles interacting with two large heat
reservoirs with $N$ particles each. The reservoirs are initially in equilibrium at temperatures $T_+$ and $T_-$, and we
assumed that the centers of mass of both the system and the reservoirs are initially at rest. We proved that the system
+ 2 reservoirs, when $N>>M$, that is the reservoirs are much larger than the system, is well approximated by a system of
$M$ particles interacting with 2 Maxwellian thermostats. This means that we can formally take the limit for $N\to\infty$
and replace the collisions of the particle in the system with particles in the reservoirs with collisions with {\it
virtual} particles randomly selected from Maxwellian distributions with temperature $T_+$ or $T_-$. Contrary to the case
of only one reservoir, see \cite{BLTV}, our results in \cite{BLP25} were only valid for a time shorter than
$\sqrt{N}/M$. Observe that, by taking the limit $N\to\infty$, we eliminate the influence of the system on the
reservoirs. This leads to the notion of a thermostat with fixed temperature which was analyzed in \cite{BLTV}. On the
other hand, for any fixed $N$, over a time of the order of $N/M$ the temperatures of the reservoir and system evolve.
That is, no matter how large (but finite) the reservoirs are, there will be temperature evolution. Thus the
approximation with thermostats at fixed temperature can only be valid for a relatively short time.

In this paper, we continue this line of study by considering a generalization of the model considered in \cite{BLP25} in
which the centers of mass of the two reservoirs initially have non-zero velocities $\vec p_+$ and $\vec p_-$, in
addition to having different temperatures. If these velocities are different, they will create a shear in the system. On
a macroscopic level, the energy stored in the center of mass motion of the reservoirs is converted into heat so that, in
this case, the final temperature of the reservoirs is higher than the average of the initial temperatures. It is
remarkable that one can write closed evolution equations for the momentum of the centers of mass and the temperatures of
the system and reservoirs, see \eqref{eq:IVPM1} and \eqref{eq:IVPT} below. As discussed in Subsection
\ref{subsec:hidro}, two time scales emerge. On a short time scale of order 1 the system reaches a quasi stationary state
while on a longer time scale of order $N/M$ the temperatures and center of mass velocities of the two reservoirs tend to
become equal and the all system + 2 reservoirs converges to equilibrium.

We can thus attempt to take the limit $N\to\infty$ in a more careful way. To create our effective evolution we assume
that the particles in the system interact with 2 infinite reservoirs whose center of mass velocities and temperatures
evolve in the same way as the corresponding quantities for the reservoirs with finite $N$. This give rise to the notion
of {\it dynamic} thermostat as introduced in Subsection \ref{subsec:dynamic} below. More precisely, we compare the
evolution of the system + 2 reservoir with that of a system interacting with 2 Maxwellian thermostat with center of mass
velocities and temperatures, which evolves in time. We believe that the notion of dynamic thermostat is the major
conceptual novelty of this paper and that thermostats of this type will be instrumental in extending our results beyond
the time range obtained here.

Our main result is contained in Theorem \ref{thm:2thermo} below. Similarly to \cite{BLP25}, we obtain that the system +
2 reservoirs is well approximated by the system + 2 dynamic thermostats for times shorter than $\sqrt N/M$. On a
technical level, as in our previous works, we use the GTW $d_2$ metric (see \cite{GTW}) to measure the distance between
the reservoir and thermostat evolution. Observe that the $d_2$ metric requires the the distributions being compared have
the same first moment. However, as already noticed, the first moment of the distributions of the system and/or the
reservoirs evolve in time so that if the first moment of a thermostat remained fixed the $d_2$-distance between the
reservoir evolution and a fixed thermostat evolution would diverge immediately. This makes the introduction of the
dynamic thermostat necessary also at a technical level, at least as far as the center of mass velocities are involved.
To prove our main theorem we follow the strategy of \cite{BLP25} based on a Duhamel expansion and on the functional
inequality introduced in \cite[Lemma 3.3]{BLTV}. After this, we are left with estimating the evolution of the
thermostated system. Due to the presence of the dynamic thermostats, this last step is considerable more involved than
the corresponding estimates in \cite{BLP25}. In this sense, Subsection \ref{subsec:1part} is the novel technical
achievement of this paper. It also makes clear that the main obstruction to a result uniform in time is the variation of
the state of the reservoirs/thermostats over the long time scale.

We also observe that if we suppose that the two reservoirs have equal temperatures  ($T_+=T_-$) and initial velocities
($\vec p_+=\vec p_-$), then the system is essentially equivalent to a system interacting with one shearing reservoir. In
this case there is no long time scale and both the shear and the temperature of the reservoir change very little along
the evolution. Thus we are again able to obtain estimates uniform in time, see Corollary \ref{coro:1thermo}. We note
that the system involving only one thermostat can be analyzed using a suitable $L^2$ norm, see \cite{BLV,BPP}. Although
the analysis is simpler, the $L^2$ norm, as discussed for example in \cite{BLV}, is not well suited for systems with
many degrees of freedom. Moreover, it is unclear how to extend these $L^2$ results to the case of a system interacting
with more than one reservoir, even for a short time.

In Subsections \ref{subsec:Kac} and \ref{subsec:res} we introduce the Kac evolution for the system and the reservoirs. 
These are included for the sake of completeness but can be skipped by the reader familiar with this family of models. 
In Subsection \ref{subsec:dynamic} we introduce the notion of dynamic thermostat and in Subsection \ref{subsec:main} we 
state our main result. Section \ref{sec:out} contains a few comments on and possible extensions of our results.
Section \ref{sec:proofs} contains the main components of the proof of Theorem \ref{thm:2thermo}. Finally, the appendices
deal with a few technical statements used in the proof.


\subsection{The Kac dynamics} \label{subsec:Kac}

The basic ingredient for the Kac model evolution is the operator which describes the effect of a collision between two
particles. Suppose that the particles are spherical and that the unit vector associated with the relative position of
their centers is $\omega$. If $\vec v_1$ and $\vec v_2$, with $\vec v_i=(v_{i,1}, v_{i,2}, v_{i,3})\in\R^3$, are the 
outgoing velocities immediately after the collision, then the incoming velocities before the
collision were 
\begin{equation}\label{eq:coll}
\vec v_1^*(\omega)=\vec v_1-\scal{(\vec v_1-\vec v_2),\omega}\omega\,,\qquad \vec v_2^*(\omega)=\vec
v_2-\scal{(\vec v_2-\vec v_1),\omega}\omega\, ,
\end{equation}
where $\scal{\vv,\vw}$ denotes the usual inner product in $\mathds R^3$. In the following, we will use $\langle \cdot,
\cdot \rangle$ to denote the usual inner product on any $\R^K$. The precise space under
consideration will be clear from the context. Suppose now that before the collision the velocities of the two particles
were distributed according to the probability distribution $g(\vec v_1,\vec v_2)$ on $\R^6$. If the unit vector $\omega$
is chosen randomly and uniformly on the unit sphere $\mathds S^2$ then the effect of the collision on $g$ is given by %
\begin{equation}\label{eq:defR}
	R[g](\vec v_1,\vec v_2)=\int_{\mathds S^2}g(\vec v_1^*(\omega),\vec v_2^*(\omega))d\omega\, .
\end{equation}
With this collision in hand, we can now start to describe the evolution in the models we will consider. 

We consider a system of $M$ particles moving in $\R^3$ with velocities $\uv=(\vv_1,\ldots,\vv_M)\in\R^{3M}$ where
$\vv_i\in\R^3$. Since the collisions occur randomly and are independent of the positions of the particles, the initial
state of the system is described by a probability distribution $f_0(\uv)$ on $\R^{3M}$. We also assume that the
particles are identical, so that $f_0$ is invariant under particle permutations. 

Contrary to our previous works, see e.g. \cite{BLP25, BLTV}, we assume that the center of mass of our system initially
moves with average velocity $\vec p_S$ given by %
\[
\vec p_S:=\frac1N\sum_i\int \vv_if_0(\uv) \ d\uv = \int \vv_1 f_0(\uv) \ d\uv \, .
\]
Since the collision described in \eqref{eq:coll} preserves the total momentum of the two particles, this assumption will
become relevant only after we introduce the reservoirs, see Subsection \ref{subsec:res} below. For the sake of notation 
we will write %
\begin{equation}\label{eq:centered}
f(\uv)=\tilde f(\vv_1-\vec p_S,\ldots,\vv_M-\vec p_S)=\tilde f(\uv-\underline p_S)
\end{equation}
where we set $\underline p_S=(\vec p_S,\ldots,\vec p_S )$. Finally we assume that the initial temperature $T_S$ of the 
system, or better the average kinetic energy per degree of freedom w.r.t. the center of mass, is finite, that is
\[
T_S:=\frac13\int_{\R^{3M}}\|\vv_1-\vec p_S\|^2 f(\uv)d\uv<\infty
\]

The effect of a collision between particle $i$ and $j$ on any distribution $f$ is described by the operator $R$
in \eqref{eq:defR} acting on $\vec v_i$ and $\vec v_j$. That is %
\begin{equation*}
	R_{i,j}^S[f](\underline v)=\int_{\mathds S^2}f(\underline v_{i,j}(\omega))d\sigma(\omega)
\end{equation*}
where 
\begin{equation*}
\underline v_{i,j}(\omega)=(\vec v_1,\ldots,\vec v_i^*(\omega),\ldots,\vec
v_j^*(\omega),\ldots,\vec v_M)\,,
\end{equation*}
with $\vec v_i^*(\omega)$ and $v_j^*(\omega)$ given by \eqref{eq:coll}.

We assume, as in Kac's original work \cite{Kac}, that the collisions take place according to a Poisson process of
intensity $\lambda_M$ and that when a collision takes place, the pair of colliding particles is chosen randomly and
uniformly. Choosing $\lambda_M=M\lambda_S/2$ we get that the infinitesimal generator for the evolution of a
system with distribution $f$ is %
\begin{equation}\label{eq:genLS}
\mathcal L_S[f]=\frac{\lambda_S}{M-1}\sum_{i<j} (R^S_{i,j}[f]-f).
\end{equation}
With the above choice of $\lambda_M$ the average number of collisions that a particle suffers in a unit time is equal to
$\lambda_S$ independently of $M$. $\lambda_S^{-1}$ is called the {\it mean free flight} for the system. 

\subsection{Shearing Reservoirs}\label{subsec:res}

We consider the situation when a system  with $M$ particles interacts with two much larger reservoirs each containing
$N$ particles with $N \gg M$. We will assume that the two reservoirs are initially in canonical equilibrium at
temperature $T_+$ and $T_-$ and that their centers of mass move with average velocities $\vec p_+$ and $\vec p_-$. We
call $\uu_+=(\vu_{+,1},\ldots, \vu_{+,N})\in \R^{3N}$ the velocities of the particles in the $+$ reservoir and similarly
$\uu_-=(\vu_{-,1},\ldots, \vu_{-,N})\in \R^{3N}$ for the particles in the $-$ reservoir. To shorten notation we will 
also use $\vu_i=(\vu_{+,i},\vu_{-,i})\in \R^6$ and $\uu=(\vu_{1},\ldots, \vu_{N})\in \R^{6N}$ for the combined 
variables of both reservoirs.

The state of the combined
system + 2 reservoirs is now given by a distribution $F(\uu,\uv)$ on $\R^{3(M+2N)}$ and the initial state is of
the form %
\begin{equation}\label{eq:ini2}
F_0(\uv,\uu)=\Gamma^N_+(\uu_+-\underline p_+)f_0(\uv)\Gamma^N_-(\uu_--\underline p_-)
\end{equation}
where, for $\sigma\in\{+,-\}$, we set $\ul p_\sigma=(\vec p_\sigma,\ldots, \vec p_\sigma)$ and
\begin{equation*}
\Gamma^N_\sigma(\uu_\sigma)=\Gamma_{T_\sigma}^{N}(\uu_\sigma)=\prod_{i=1}^{N} \Gamma_{T_\sigma}(\vu_{\sigma,i})
\end{equation*}
with
\begin{equation*}
\Gamma_{T}(\vv)=\left(2\pi T\right)^{-\frac32}e^{-\frac {\|\vv\|^2}{2T}}\, .
\end{equation*}
We can write the generator of the combined system + 2 reservoirs evolution as
\begin{equation}\label{eq:gen2res}
\mathcal L=\mathcal L_S+\mathcal L_{R_+}+\mathcal L_{R_-}+\mathcal L_{I_+}+\mathcal L_{I_-}
\end{equation}
where $\mathcal L_S$ still describes the internal evolution of the system, see \eqref{eq:genLS}, while $\mathcal
L_{R_\sigma}$ and $\mathcal L_{I_\sigma}$, with $\sigma\in\{+,-\}$, describe the internal evolution of the $\sigma$ reservoir and
the interaction of the $\sigma$ reservoir with the system, respectively. This means that
\begin{equation*}
\mathcal L_{R_\sigma}[F]=\frac{\lambda_{R}}{N-1}\sum_{i<j} (R_{i,j}^{R_\sigma}[F]-F)
\end{equation*}
with $R_{i,j}^{R_\sigma}$ describing the effect of a collision in the $\sigma$ reservoir between particle $i$ with 
velocity $\v
u_{\sigma,i}$ and particle $j$ with velocity $\v u_{\sigma,j}$ via the operator \eqref{eq:defR} acting on $\v
u_{\sigma,i}$ and $\v u_{\sigma,j}$. Finally $\mathcal L_{I_\sigma}$ describes the interaction between
the $\sigma$ reservoir and the system. It is also modeled via a Kac-style collision and we set
\begin{equation}\label{eq:genLI}
\mathcal L_{I_\sigma}[F]=\frac \mu N\sum_{i=1}^N\sum_{j=1}^M (R_{i,j}^{I_\sigma}[F]-F)
\end{equation}
where $R_{i,j}^{I_\sigma}$ describes the effect of a collision between particle $i$ in the $\sigma$ reservoir with 
velocity $\v
u_{\sigma,i}$ and particle $j$ in the system with velocity $\v v_j$. It is again given by the operator \eqref{eq:defR}
acting on $\v u_{\sigma,i}$ and $\vec v_j$. The factor $\mu/N$ in \eqref{eq:genLI} is chosen so that a particle in the
system suffers, on average, $\mu$ collisions per unit time with particles in the reservoir. At the same time, a particle
in the reservoir on average suffers $\mu M/N$ collision per unit time with particles in the system. Observe that there
is no direct interaction between the two reservoirs. 

The state of the system + 2 reservoirs at time $t$ is thus given by
\begin{equation}\label{eq:evolres}
F_t=e^{\cL t}F_0\, .
\end{equation}
For every $N$ and $M$, $\cL$ is a bounded operator on $L^1(\R^{3M+2N})$ so that the evolution \eqref{eq:evolres} is 
well defined for every $t$. Finally,
for simplicity of notation, we choose units of time such that $\mu=1$. 


\subsection{Dynamic Shearing Thermostats}\label{subsec:dynamic}

In the spirit of \cite{BLP25}, we want to obtain an effective approximation for the evolution of the system + 2
reservoirs considered in the previous subsection when $N$ is very large. To do this we first need a better understanding
of its macroscopic evolution. We can define the average velocities of the centers of mass at time $t$ as
\begin{equation*}
\begin{aligned}
\vec m_{S}(t) &:= \frac 1M\sum_{i=1}^N\int \vv_i F_t( \uv, \uu) \  d\uv d\uu =\int \vv_1 F_t( 
\uv, \uu) \  d\uv d\uu\\
\vec m_{\sigma}(t) &:= \frac 1N\sum_{i=1}^N\int \vu_{\sigma,i} F_t(\uv, \uu) d\uv d\uu = \int 
\vu_{\sigma,1} F_t(\uv, \uu) \  d\uv d\uu\,. 
\end{aligned}
\end{equation*}
For simplicity sake, we will call $\vec m_S$ the {\it velocity of the system} and similarly for $\vec m_\sigma$.  
Form the properties of the generator $\cL$, see \eqref{eq:gen2res}, it is not hard to derive explicit closed equations 
for the evolution of $\vec m_{S}$ and $\vec m_{\sigma}$, see \cite{BLP25} for more details. Indeed we get the Initial 
Value Problem
\begin{equation}\label{eq:IVPM1}
\frac{d}{dt}\begin{pmatrix}
		\vec m_+ \\ \vec m_S \\ \vec m_-
	\end{pmatrix}=
\frac 13\begin{pmatrix}
-\frac{M}{N} & \frac MN & 0\\
1 & -2 & 1\\
0 & \frac{M}{N} & -\frac MN
\end{pmatrix}
\begin{pmatrix}
		\vec m_+ \\ \vec m_S \\ \vec m_-
	\end{pmatrix},
\qquad\qquad
\begin{pmatrix}
		\vec m_+(0) \\ \vec m_S(0) \\ \vec m_-(0)
	\end{pmatrix}=
\begin{pmatrix}
		\vec p_+ \\ \vec p_S \\ \vec p_-
	\end{pmatrix}
\end{equation}
To simplify the coming discussion we will chose a reference frame in which the combined center of mass of the system + 2
reservoirs is initially at rest, that is $N\vec p_++N\vec p_-+M\vec p_S=0$. This condition is clearly preserved in time.

Similarly, the temperatures, or better the average kinetic energies per degree of freedom w.r.t. the centers of mass, at
time $t$ can be defined as
\begin{equation*}
\begin{aligned}
\tau_S(t):=&\frac 13\int \|\vv_1-\vec m_S(t)\|^2 F_t(\uv, \uu) \ d\uv d\uu\\
\tau_\sigma(t):=&\frac 13\int \|\vu_{\sigma,1}-\vec m_\sigma(t)\|^2 F_t( \uv, \uu) \  d\uv d\uu\, .
\end{aligned}
\end{equation*}
and they satisfy the IVP
\begin{equation}\label{eq:IVPT}
\frac{d}{dt}\begin{pmatrix}
		\tau_+ \\ \tau_S \\\tau _-
	\end{pmatrix}=
\frac 13\begin{pmatrix}
-\frac{M}{N} & \frac MN & 0\\
1 & -2 & 1\\
0 & \frac{M}{N} & -\frac MN
\end{pmatrix}
\begin{pmatrix}
		\tau_+ \\ \tau_S \\\tau _-
\end{pmatrix} + 
\frac 19\begin{pmatrix}
\frac{M}{N} \|\vec m_+-\vec m_S\|^2\\
\|\vec m_+-\vec m_S\|^2+ \|\vec 
m_S-\vec m_-\|^2\\
\frac{M}{N} \|\vec m_S-\vec m_-\|^2
\end{pmatrix},\qquad
\begin{pmatrix}
		\tau_+(0) \\ \tau_S(0) \\ \tau _-(0)
	\end{pmatrix}=
\begin{pmatrix}
		T_+ \\ T_S \\ T _-
	\end{pmatrix}
\end{equation}
Let $\vec m_\sigma(t)$ and $\vec m_S(t)$ be the solutions of \eqref{eq:IVPM1} and let $\underline m_S(t) = (\vec
m_S(t),...,\vec m_S(t)) \in \R^{3M}$ and $\underline m_\sigma(t) = (\vec m_\sigma(t),...,\vec m_\sigma(t)) \in \R^{3N}.$
Similarly let $\tau_S(t)$ and $\tau_\sigma(t)$ be the solution of \eqref{eq:IVPT}. In Subsection \ref{subsec:hidro}
below we will give a more detailed discussion of these equations, their explicit solutions, and their physical meaning.
We observe here that
\begin{equation}\label{eq:infty}
\begin{aligned}
\lim_{t\to\infty}\vec 
m_S(t)&=\lim_{t\to\infty}\vec m_\pm(t)=0\\
\lim_{t\to\infty}\tau_S(t)&=\lim_{t\to\infty}\tau_\pm(t)=\frac{N(T_++T_-)+MT_S}{N^*}+\frac{N\|\vec m_+-\vec m_S\| + 
N\|\vec m_+-\vec m_S\|}{3N^*} =:\tau_\infty. 
\end{aligned}
\end{equation}
where $N^* = 2N + M$ is the total number of particles. This means that the system +2 reservoirs eventually reaches an 
equilibrium state where the the system and the 2 reservoirs are at rest and have the same temperature.

We are now ready to introduce our effective evolution. The following heuristic discussion will lead us to the evolution 
equation \eqref{eq:evolther} for the system interacting with 2 {\it dynamic} thermostats. The validity of the 
approximations discussed below will be proved {\it a posteriori} in Theorem \ref{thm:2thermo}, at least for a not too 
long time

From \eqref{eq:IVPM1} and \eqref{eq:IVPT}, see also Subsection \ref{subsec:hidro}, we see that $\vec m_\sigma(t)$ and
$\tau_\sigma(t)$ vary slowly in time with a rate of the order of $M/N$.  On the other hand, the Kac collisions between
particle in the same reservoir tend to keep the state of the reservoir close to a distribution that is rotational
invariant in the direction orthogonal to the total momentum. Finally, since $N$ is very large and the initial state of
each reservoir is a product state, we expect that the Kac collision inside the reservoirs will tend to push the state of
the reservoirs toward distributions which are chaotic, in the sense that they remain close to product states. See 
\cite{Kac} for more details on chaotic distributions.

From the above discussion, we expect that %
\begin{equation}\label{eq:appevo}
F_t(\uv,\uu)\simeq G_{+,t}^N(\uu_+)f_t(\uv)G_{-,t}^N(\uu_-)
\end{equation}
where for $\sigma\in\{+,-\}$ we set
\begin{equation}\label{eq:Gsigma}
G^N_{\sigma,t}(\uu_\sigma):=\prod_{i=1}^{N}G_{\sigma,t}(\uu_{\sigma,i})\qquad\hbox{with}\qquad 
G_{\sigma,t}(\uu_{\sigma,i}):= 
\Gamma_{\tau_\sigma(t)}(\vu_{\sigma,i}-\vec m_\sigma(t))
\end{equation}
with $\vec m_S(t)$, $\vec m_\sigma(t)$ solution of \eqref{eq:IVPM1} and $\tau_S(t)$ and $\tau_\sigma(t)$ solutions of
\eqref{eq:IVPT}. That is, we think that the reservoirs evolve slowly in time and remain always close to the canonical
state defined by the solution of \eqref{eq:IVPM1} and \eqref{eq:IVPT}.

In light of this, we are led to conjecture that the evolution of the system can be effectively described via the
interaction with two dynamic thermostats. That is, if the state of the system at time $t$ is $f_t$, we consider the
state
\[
\widetilde F_t( \uv, \uu) := G_{+,t}^N(\uu_+)f_t(\uv)G_{-,t}^N(\uu_-)
\]
formed by the system plus the two reservoirs in their effective state. We then compute the time variation of $f_t$ by 
applying the generator $\cL$, see \eqref{eq:gen2res}, and take the marginal over the variables of the reservoirs. 
That is, the evolution equation for $f_t$ becomes:
\begin{equation}\label{eq:ther}
\frac{d}{dt} f_t(\uv) = \int \mathcal L [\widetilde F_t](\uv, \uu) \ d\uu\, .
\end{equation}
This means that the system feels the influence of two reservoirs whose state at time $t$ is prescribed a priori by 
\eqref{eq:Gsigma} and thus is not influenced by the state of the system. To be more precise, we define the 
one-parameter family of 
operators
\begin{equation*}
\Psi(t)[f]  = \int \mathcal L[G_{+,t}^N f G_{-,t}^N] (\uu_+, \uv, \uu_-) \ d\uu_+ d\uu_+
\end{equation*}
and the two-parameter family $\Phi_{t,s}$, $t \geq s$, that satisfies
\begin{equation}\label{eq:2times}
\frac{d}{dt} \Phi_{t,s} = \Psi(t) \Phi_{t,s},\qquad\Phi_{t,t}=\mathrm{Id}
\end{equation}
Thus the solution of \eqref{eq:ther} can be written as
\begin{equation}\label{eq:evolther}
f_t=\Phi_{t,0}f_0\, 
\end{equation}
and hence the evolution of the system + 2 thermostats is given by
\begin{equation*}
\widetilde F_t(\underline v, \underline u) = G_{+,t}^N(\underline u_+)\Phi_{t,0} [f_0](\underline v)G_{-,t}^N(\underline u_-)
\end{equation*}
As for $\cL$, it is not hard to see that $\Psi(t)$ is a bounded operator for every $t$ so that $\Phi_{t,s}$ is well 
defined for every $s\leq t$. Moreover we can define
\begin{equation}\label{eq:tildem}
\vec{\tilde m}_S(t):=\int \vec v_{1}  \widetilde F_t(\uv,\uu) d\uv d\uu\qquad\hbox{and}\qquad
\tilde\tau_S(t):=\frac13\int \|\vec v_{1}-\vec {\tilde m}_S(t)\|^2 \widetilde F_t(\uv,\uu) d\uv d\uu
\end{equation}
and similar definitions for $\vec{\tilde m}_\sigma(t)$ and $\tilde \tau_\sigma(t)$. Notice that $\vec{\tilde m}_S$ and
$\tilde\tau_S$ satisfy the same differential equation as $\vec{m}_S$ and $\tau_S$ in \eqref{eq:IVPM1} and
\eqref{eq:IVPT} with the same initial conditions while $\vec{\tilde m}_\sigma(t)$ and $\tilde \tau_\sigma(t)$ are
clearly identical with $\vec{m}_\sigma(t)$ and $\tau_\sigma(t)$. It thus follow that
\[
\vec{\tilde m}_S(t)=\vec{m}_S(t)\qquad\hbox{and}\qquad\tilde\tau_S(t)=\tau_S(t)
\]
for every $t$, that is the macroscopic evolution of the system + 2 thermostats coincides with that of the system + 2 
reservoirs. Observe finally that, due to the product nature of $G^N_{\sigma,t}$, $\Psi(t)$ does not depend on $N$ and 
it can be written in a form very similar to the {\it static} thermostat studied in \cite{BLV,BLTV,BLP25}, see 
Subsection \ref{subsec:statdyn} for more details.

In the case in which $T_+=T_-=:T$ and $\vec p_+=\vec p_-=:\vec p$, we can imagine that the two reservoirs constitute a
single reservoir with $2N$ particles at temperature $T$ and velocity $\vec p$. Moreover, if we fix the initial velocity
of the system + reservoir to be 0 we get $\vec p=-\frac{M}{2N}\vec p_S$. Thus, if $N>>M$, we can imagine that the
reservoir is effectively at rest and interacts with a system that is initially in motion with velocity $\vec p_S$. It is
easy to see from \eqref{eq:IVPM1} and \eqref{eq:IVPT} that in this situation the temperature and the velocity of the
reservoir vary very little along the evolution. For these reasons, in this particular case, we will be able to obtain
estimates that are uniform in time, see Corollary \ref{coro:1thermo} below.

\subsection{Main Result: comparing $F_t$ and $\widetilde F_t$}\label{subsec:main}

We want to compare the system + 2 reservoir evolution $F_t$ given by \eqref{eq:evolres} with the system + two thermostat
evolution $\widetilde F_t$ given by \eqref{eq:appevo} with $f_t$ defined by \eqref{eq:evolther}. As in \cite{BLP25}, we
will compare $F_t$ and $\widetilde F_t$ using the GTW $d_2$ distance, see \cite{GTW}, that we briefly introduce and 
discuss below. 

For $f:\R^{K}\mapsto \R$ we write its Fourier transform as
\begin{equation*}
\hat f(\underline \xi)=\int_{\R^K}e^{2\pi i \scal{\underline \xi,\underline v}}f(\uv)d\uv.
\end{equation*}
Given two distributions $f,g:\R^{K}\mapsto \R$ with
\begin{equation}\label{eq:cond}
\int_{\R^K}f(\uv)d\uv=\int_{\R^K}g(\uv)d\uv=1
\qquad\hbox{and}\qquad 
\int_{\R^K}\uv f(\uv)d\uv=\int_{\R^K}\uv g(\uv)d\uv
\end{equation}
we define the $d_2$ distance between $f$ and $g$ as
\begin{equation}\label{eq:d2}
d_2(f,g)=\sup_{\uxi\not=0}\frac{|\hat f(\uxi)- \hat g(\uxi)|}{\|\uxi\|^2}.
\end{equation}
Notice that the two conditions in \eqref{eq:cond} are necessary for the supremum in \eqref{eq:d2} to be finite. On the 
other hand if we further assume that
\[
\int_{\R^K}\|\uv\|^2 f(\uv)d\uv,\int_{\R^K}\|\uv\|^2 g(\uv)d\uv<\infty
\]
then we get that $d_2(f,g)<\infty$. Finally we observe that the $d_2$ distance is unchanged if we change the average
velocity. That is, for any $\underline m\in \R^K$ we have %
\begin{equation}\label{eq:shift}
d_2(f(\cdot-\underline m), g(\cdot-\underline m))=d_2(f,g)\, .
\end{equation}
Thus the choice of a particular reference frame, as discussed after \eqref{eq:IVPM1}, has no effect on our results.

Like our result in \cite{BLP25}, in the case of two thermostat we are unable to obtain a uniform estimate in $t$ due to
the fact that the states of the reservoirs change in time. Notwithstanding the temperatures and velocities of our
dynamic thermostats exactly follow the temperatures and velocities of the reservoirs,  this is apparently not enough.
See Subsection \ref{subsec:hidro} for a more detailed discussion of what we think is needed to obtain a uniform
estimate.

We formulate our results in terms of the distance $d_2(\tilde f_0,\Gamma^M_{\tau_\infty})$ between the recentered
initial distribution $\tilde f_0$, see \eqref{eq:centered}, and the Maxwellian at equilibrium temperature $\tau_\infty$,
together with the two quantities that characterize how far from equilibrium the reservoirs are at time $t=0$:
\begin{equation}\label{eq:defp}
\Delta_p=\|\vec p_+-\vec p_S\|^2+\|\vec p_--\vec p_S\|^2\qquad\hbox{and}\qquad
\Delta_T=|T_+-T_S|+|T_--T_S|\, .
\end{equation}
We are now ready for our main Theorem.

\begin{mythm}\label{thm:2thermo}
Let $f_0$ be a distribution on $\R^{3M}$, $T_+,T_-\in\R$ and $\vec p_+, \vec p_-\in \R^3$ and consider the initial state
$F_0$ given by \eqref{eq:ini2} and its evolution $F_t$ given by \eqref{eq:evolres}. Consider also the evolution of the
initial state $F_0$ given by $\widetilde F_t=G_{+,t}^N\Phi_{t,0} [f_0]G_{-,t}^N$ with $\Phi_{t,0}$ defined in
\eqref{eq:2times}. Then we have

\begin{equation}\label{eq:mainth}
d_2(F_t, \widetilde F_t) \leq C(f_0)\frac M{\sqrt N}\left(d_2(\tilde f_0, \Gamma^M_{\tau_\infty})^{\frac 
16}\left(1-e^{-\frac t{18}}\right)+(\Delta_p^2+\Delta_T)^{\frac 16}t\right)
\end{equation}
for all $t\geq 0$ and a suitable constant $C(f_0)$.
\end{mythm}

\begin{Remark}
\emph{The constant $C(f_0)$ is \eqref{eq:mainth} does not depend on $N$ or $M$ but it still depends on $f_0$ through its
moments of low degree, including $\vec p_S$ and $T_S$, together with $T_\sigma$ and $\vec p_\sigma$, see Lemma
\ref{lem:prop5new} and Subsection \ref{subsec:1part} below for more details. We chose to explicitly write only the
dependence on quantities that are relevant for the non equilibrium character of our model. }
\end{Remark} 

In the case that the two thermostats have the same initial temperature and the same velocity we can obtain a better 
result.

\begin{corollary}\label{coro:1thermo}
In the hypotheses of Theorem \ref{thm:2thermo}, if $T_+=T_-=:T$ and $\vec p_+=\vec p_-$ we get
\[
d_2(F_t, \widetilde F_t) \leq  C\frac M{\sqrt N}\left(d_2(f_0, \Gamma_T^M)^{\frac 16}+(\Delta_p^2+\Delta_T)^{\frac 
16}\right)\left (1-e^{-\frac t{18}}\right)
\]
for all $t>0$.
\end{corollary} 

Observe that, even if the two reservoirs start with $T_+=T_-$ but $\vec p_+\not =\vec p_-$, the temperature of the
reservoirs will change with time, see \eqref{eq:IVPM1} and discussion in Subsection \ref{subsec:hidro}. Thus both the
conditions on the temperatures and velocities are necessary for the validity of Corollary \ref{coro:1thermo}.

\section{Discussion and Outlook}\label{sec:out}

We collect here few comments on the meaning of our results and suggestions for possible extensions.

\subsection{Hydrodynamics}\label{subsec:hidro}

It is not hard to see that the solution to \eqref{eq:IVPM1} can be written as:
\begin{equation}\label{eq:IVPsol}
	\begin{pmatrix}
		\vec m_+(t) \\ \vec m_S(t) \\ \vec m_-(t)
	\end{pmatrix}=
\frac{\vec p_+-\vec p_-}2\begin{pmatrix}
	1 \\ 0 \\ -1
\end{pmatrix} e^{-\frac M{3N} t}+
\frac{\vec p_+-2\vec p_S+\vec p_-}{2N^*}\begin{pmatrix}
	M \\ -2N \\ M
\end{pmatrix}e^{-\frac 13\left(2+\frac MN\right)t}
\end{equation}
where we have used our assumption that the total center of mass is initially at rest. Thus, on a time scale of order 1,
$\vec m_S(t)$ goes to 0 while $\vec m_+(t)$ and $\vec m_-(t)$ become equal and opposite. On a longer time scale of order
of $N/M$,  $\vec m_\sigma(t)$ also tends to zero.

Inspection of equations \eqref{eq:IVPM1} and \eqref{eq:IVPT} shows that the loss of kinetic energy by the center of mass
motion of the reservoirs and system is compensated by the positive contribution to the temperature of the non
homogeneous terms in \eqref{eq:IVPT}. Indeed, we can see from \eqref{eq:IVPsol} that $\vec m_S(t)$ approaches its
asymptotic value exponentially fast with a rate of order 1. The kinetic energy-temperature correspondence discussed
above (or, indeed, direct inspection of \eqref{eq:IVPT}) leads us to see that $\tau_S(t)$ also approaches its asymptotic
value exponentially fast with a rate of order 1. Thus, to understand the large time behavior of the solution of
\eqref{eq:IVPT}, we consider the special case in which $\vec p_S=0$ so that $\vec p_+=-\vec p_-$ while
$T_S=(T_++T_-)/2$. In this situation we can write the solution of \eqref{eq:IVPT} as
\begin{align*}
\begin{split}
\tau_+(t)=&T_S+\frac {2N}{3N^*}\|p_+\|^2+\frac{T_+-T_-}2e^{-\frac M{3N} t}+\frac{M}{3N^*}\|\vec p_+\|^2e^{-\frac 
13\left(2+\frac 
MN\right)t}-\frac{\|p_+\|^2}3e^{-\frac {2M}{3N} t}\\
\tau_S(t)=&T_S+\frac {2N}{3N^*}\|p_+\|^2-\frac {2N}{3N^*}\|p_+\|^2e^{-\frac 13\left(2+\frac 
MN\right)t}\\
\tau_-(t)=&T_S+\frac {2N}{3N^*}\|p_+\|^2-\frac{T_+-T_-}2e^{-\frac M{3N} t}+\frac{M}{3N^*}\|\vec p_+\|^2e^{-\frac 
13\left(2+\frac 
MN\right)t}-\frac{\|p_+\|^2}3e^{-\frac {2M}{3N} t}\\
\end{split}
\end{align*}
where we clearly see that the final temperature of the system and reservoirs is greater than $T_S$ due to the
transformation of the kinetic energy of the centers of mass of the reservoirs into heat thanks to the {\it viscosity}
terms in \eqref{eq:IVPT}. Observe that if $T_+=T_-$, after an initial transient, the temperature of the system is above
the temperatures of the two reservoirs. That is, the heat produced by the friction between the two reservoirs due to
their opposite motion heats up the system and then is diffused back to the reservoirs until equilibrium is reached.

\subsection{Static vs dynamic thermostats}\label{subsec:statdyn}

To make a connection with the definition of the thermostat in \cite{BLTV,BLP25}, we observe that for 
$g:\R^3\mapsto\R$ we can define
\begin{equation}\label{eq:B}
B_\sigma(t)[g](\vv)=\int_{\R^3}\int_{\mathds S^2} \Gamma_{\tau_\sigma(t)}(\vu^*_\sigma-\vec m_\sigma(t))g(\vv^*)d\omega 
d\vu_\sigma=
\int_{\R^3} R[\Gamma_{\tau_\sigma(t)}(\cdot-\vec m_\sigma(t))g](\vu_\sigma,\vv)d\vu_\sigma
\end{equation}
For $f:\R^{3M}\mapsto \R$, we can call $B_{k,\sigma(t)}[\tilde f]$ the operator \eqref{eq:B} acting on $\vv_k$,
which is precisely the form of the thermostat defined in \cite{BLTV,BLP25}. Following \cite{CLM}, we can combine
the two thermostats and define
\begin{equation}\label{eq:BM}
B(t)[f](\vu):=\int_{\R^3} R[\mathcal M_tf](\vu,\vv)d\vu
\end{equation}
where
\[
\mathcal M_t(\vu):=\frac 12\left(\Gamma_{\tau_+(t)}(\vu-\vec m_+(t))+\Gamma_{\tau_-(t)}(\vu-\vec 
m_-(t))\right)
\]
and obtain
\begin{equation}\label{eq:oldPsi}
\Psi(t)[f] = \cL_S + 2\sum_{k=1}^M (B_{k}(t)[f] - f):=\cL_S[f] + \cL_{B(t)}[f]\, .
\end{equation}
where, as usual, $B_k(t)$ is the operator $B(t)$ acting on the $k$-th variable $\vv_k$. We observe that
\eqref{eq:oldPsi} shows that $\Psi(t)$ does not depend on $N$. Moreover we only need $\mathcal M(t)$ to be a probability
distribution for $\cL_{B(t)}[f]$ to define a generator on the space of distributions.

In \cite{BLP25} we studied the system + 2 reservoirs described in Subsection \ref{subsec:res} with the difference that
we assumed $\vec p_\sigma=\vec p_S=0$ and we compared its evolution with that of a system + 2 {\it static} thermostats,
that is two thermostats at rest and with temperature fixed at its initial value. We observe that in the present paper we
could not have used static thermostat. Indeed  for the $d_2$ distance between $F_t$ and $\widetilde F_t$ to be finite we
need $\vec {\tilde m}_S(t)=\vec m_S(t)$ and $\vec {\tilde m}_\sigma(t)=\vec m_\sigma(t)$, see \eqref{eq:cond} and
\eqref{eq:tildem}.

On the other hand, we could have used dynamic thermostats with evolving velocity but fixed temperatures equal to the
initial temperature of the reservoir. This would have lead to a result essentially identical to Theorem
\ref{thm:2thermo} with minor changes in the proof, see in particular \eqref{eq:d2long} and \eqref{eq:dottau}. As already
discussed in \cite{BLP25}, we cannot hope to obtain a uniform estimate for the distance between $F_t$ and $\widetilde
F_t$ without having a dynamic thermostat where the temperatures also evolve in time. On the other hand, the present
paper shows that this is not enough in itself. Indeed, we can see that the unbounded term in Theorem
\ref{thm:2thermo} arises from the presence of a long time scale in the comparison between the evolution of the
reservoirs and the evolution of the thermostats, see \eqref{eq:slow} and \eqref{eq:fullt2}. To obtain uniform in time
estimates, we would need to improve \eqref{eq:final} to an order of $1/N^{1 + \epsilon}$, for some $\epsilon>0$, where
it now reads $1/\sqrt{N}$. 
This is, actually, rather natural. Indeed, the thermostat we consider here is effectively a zeroth order in $1/N$
approximation of the Kac evolution, even with the moving thermostats. Hence, from an analytic point of view, using
dynamic thermostats is not enough. In order to improve the estimates, we need to add to the dynamic thermostat a
first-order approximation of the Kac evolution, which would probably arise from an analysis of the associated
Boltzmann-Kac equation.

Observe finally that for any fixed $t$, the equation $\Psi(t)[\phi_t]=0$ has a unique normalized solution. This is the
non equilibrium steady state the for system + 2 thermostats at fixed velocities $\vec m_\sigma(t)$ and temperatures
$\tau_\sigma(t)$. Since, for $N$ large, $\vec m_\sigma(t)$ and $\tau_\sigma(t)$ evolve very slowly, it is natural to
think that $f_t\simeq \phi_t$. That is, on a time scale much longer than $N/M$, the reservoirs evolve along the
canonical states determined by the solution of \eqref{eq:IVPM1} and \eqref{eq:IVPT} while the sytem evolve along a
slowly varying sequence of non equilibrium steady states $\phi_t$. This is the effective picture that we hope to
establish in our forthcoming works.

\medskip

\section{Proof of the Theorem \ref{thm:2thermo}}\label{sec:proofs}

We will begin with the same general set-up as in \cite{BLP25} and adjust the argument where necessary to account for
the evolving thermostats.

\subsection{Step 1: Duhamel expansion}\label{subsec:duhamel}

Our goal here is to reduce the computation of $d_2(F_t, \widetilde F_t)$ to a point where we may appeal to the analysis
we performed in \cite{BLP25}. As in \cite{BLP25}, we begin with a Duhamel expansion that yields %
\begin{align*}
\begin{split}
F_t - \widetilde F_t &= \int_0^t e^{\mathcal L(t - s)} \left[\mathcal L [\widetilde F_s] - \frac{d}{ds} \widetilde 
F_s\right] \ ds\\
&= \int_0^t e^{\mathcal L(t - s)} \left[\mathcal L [\widetilde F_s] - \Psi(s)[f_s] G_{+,s}^N G_{-,s}^N - f_s 
\frac{d}{ds}\left(G_{+,s}^NG_{-,s}^N\right) \right] \ ds\\
\end{split}
\end{align*}
Next, we pass the above into Fourier space. Let $\uxi=(\vxi_1,\ldots\vxi_N)$ be the Fourier variables associated with 
$\uv$ while $\ueta_\sigma=(\veta_{\sigma,1},\ldots,\veta_{\sigma,1})$ are the Fourier variables associated
with $\uu_\sigma$. We will also write $\ueta=(\veta_1,\ldots,\veta_N)$ with $\veta_i=(\veta_{+,1},\veta_{-,1})$ for the 
combined Fourier variables of the two reservoirs. We 
obtain
\begin{align*}
\widehat F_t - \widehat {\widetilde F}_t = \int_0^t e^{\mathcal L(t - s)} \left[\mathcal L [\widehat {\widetilde F}_s] -
\widehat\Psi(s) [\hat{f}_s] \widehat G_{+,s}^N\widehat G_{-,s}^N - \sum_{\sigma\in\{+,-\}}\left(2\pi
i\scal{\ueta_\sigma, \dot{\underline m}_\sigma(s)}+\frac 12\|\ueta_\sigma\|^2\dot\tau_\sigma(s)\right)\widehat
{\widetilde F}_s\right] \ ds
\end{align*}
where we have used that $\widehat {\mathcal L}=\mathcal L$. Moreover, for $\sigma\in\{+,-\}$,
\[
\widehat G_{\sigma,s}^N(\ueta_\sigma) =\prod_{i=1}^N\widehat G_{\sigma,s}(\veta_{\sigma,i}) \qquad\hbox{with}\qquad 
\widehat G_{\sigma,s}(\veta)=e^{-2\pi i\langle \veta, \vec m_\sigma(t)\rangle} e^{-\frac 
{\tau_\sigma(s)}{2} \|\veta\|^2}
\]
is the Fourier transform of $ G_{\sigma,s}^N$ while $\widehat\Psi(s)$ is the Fourier transform of $\Psi(s)$ and it
is given  by
\[
\widehat\Psi(s)[\hat{f}](\uxi)=\mathcal L[\widehat G_{s,-}^N\hat{f}\widehat G_{+,s}^N]( \uxi,0)\, .
\]
Thus we obtain
\begin{equation}\label{eq:d2long}
\begin{aligned}
d_2(F_t, \widetilde F_t) \leq 
\sup_{\underline \xi, \underline \eta \ne 0} \frac{1}{\|\underline \xi\|^2 + \|\underline\eta\|^2 }\int_0^t ds
\bigg|&\mathcal L [\widehat {\widetilde F}_s](\uxi,\ueta) - 
\widehat\Psi(s) [\hat{f}_s](\uxi) \widehat G_{+,s}^N(\ueta_+)\widehat G_{-,s}^N(\ueta_-) - \\
& \sum_{\sigma\in\{+,-\}}\left(2\pi i\scal{\ueta_\sigma, \dot 
{\underline m}_\sigma(s)}+\frac 12\|\ueta_\sigma\|^2\dot \tau_\sigma(s)\right)\widehat {\widetilde
F}_s(\uxi,\ueta)\bigg| .\end{aligned}
\end{equation}
We first observe that
\begin{equation}\label{eq:dottau}
\sup_{\underline \xi, \underline \eta \ne 0} \frac{ \sum_{\sigma\in\{+,-\}}\|\ueta_\sigma\|^2|\dot 
\tau_\sigma(s)|\widehat {\widetilde
F}_s(\ueta_-,\uxi,\ueta_+)}{\|\underline \xi\|^2 + \|\underline\eta\|^2 }\leq \frac {M}{3N}\sum_{\sigma\in\{+,-\}}
\left( |\tau_\sigma(s) - \tau_S(s)|+\frac{1}{3} \|\vec m_\sigma(s)-\vec m_S(s)\|^2\right)
\end{equation}
where we have used \eqref{eq:IVPT}. To bound the remaining terms in the r.h.s of \eqref{eq:d2long}, mirroring our
approach from \cite{BLP25}, we can expand $\mathcal L$ into a suitable sum involving the collision operators
$R^{I_\sigma}_{i,j}$. Since each collision only sees two of the particles at a time, it is convenient to `factor' the
remaining particles out of the collision to simplify the expression, so we begin by setting
$\ueta_\sigma^k=(\veta_{\sigma,1},\ldots,\veta_{\sigma,k-1},\veta_{\sigma,k+1},\ldots,\veta_{\sigma,N})$ and
$\ueta^k=(\veta_{1},\ldots,\veta_{k-1},\veta_{k+1},\ldots,\veta_{N})$. Clearly, $\underline \eta^k$ is simply
$\underline \eta$ with the $k$th component, $\veta_{k}$, removed. 
For convenience of notation, we also define the operator $H_s$: %
\begin{equation}\label{eq:H}
\begin{aligned}
H_{s}[\hat f](\underline \xi, \vec\eta_{k}) =\sum_{j = 1}^M\sum_{\sigma\in\{\pm\}}&\left( R^{I_\sigma}_{k,j}[\hat 
{f}\widehat G_{\sigma,s}^1](\underline \xi, \veta_{\sigma,k}) - 
R^{I_\sigma}_{k,j}[\hat 
{f}\widehat G_{\sigma,s}^1](\underline \xi, 
0)G_{\sigma,s}^1(\veta_{\sigma,k})\right)G_{\sigma',s}^1(\veta_{\sigma',k})-\\
& 2\pi i N \widehat G_{+,s}^1(\veta_{+,k})\widehat G_{-,s}^1(\veta_{-,k})\sum_{\sigma\in\{+,-\}}\langle\vec 
\eta_{\sigma,k}, \dot{\vec m}_\sigma(s)\rangle\hat {f}(\underline \xi)
\end{aligned}
\end{equation}
where $\sigma' = -\sigma$.
With this, we can indeed `factor' the integrand above and obtain
\begin{equation*}
\begin{aligned}
\mathcal L [\widehat {\widetilde F}_s](\uxi,\ueta) - 
\widehat\Psi(s) [\hat{ f}_s](\uxi) \widehat G_s^+(\ueta_+)\widehat G_s^-(\ueta_-) - &2\pi 
i\sum_{\sigma\in\{\pm\}}\langle\ueta_\sigma, \dot{\underline m}_\sigma(s)\rangle\widehat {\widetilde
F}_s(\uxi,\ueta)=\\
&\frac 1 N  \sum_{k = 1}^N H_{s}[\hat{ f}_s](\vec \xi_j, \vec\eta_{k}) \widehat 
G_{+,s}^{N-1}(\ueta_+^k)\widehat G_{-,s}^{N-1}(\ueta_-^k)
\end{aligned}
\end{equation*}
\begin{Remark}\emph{
Note that the first line in \eqref{eq:H} is, at first sight, of order $M$ while the second line apparently behaves
differently (and in a problematic way) as a function of $N$ and $M$. Closer inspection, however, reveals this is not the
case since, from \eqref{eq:IVPM1}, we see that $\frac{d}{ds}\vec m_\sigma(s)=\frac{M}{3N}(\vec m_S(s) - \vec
m_\sigma(s))$ and thus the second line in \eqref{eq:H} is also of order $M$. }
\end{Remark}

We can now define for any 
$K \geq 1$
\begin{equation}\label{eq:DN}
\mathcal D_K(H, g,\uxi):=\sup_{\ueta\not=0}\frac{\left|\sum_{i=1}^K 
H(\uxi,\veta_{i})\prod_{j\not=i}^Kg(\veta_{j})\right|}{\|\uxi\|^2 + \|\ueta\|^2} \, .
\end{equation}
so that setting $\widehat G_{s}^2(\veta_k)= \widehat G_{+,s}(\veta_{+,k})\widehat G_{-,s}(\veta_{-,k})$ we can 
finally write
\[
\begin{aligned}
\sup_{\uxi, \ueta \ne 0} 
\frac1{\|\underline \xi\|^2 + \|\underline\eta\|^2}\biggl|\mathcal L [\widehat {\widetilde 
F}_s](\ueta_-,\uxi,\ueta_+) 
- &\widehat\Psi(s) [\hat{f}_s)](\uxi) \widehat G_{+,s}^N(\ueta_+)\widehat G_{-,s}^N(\ueta_-) - \\
&2\pi i\sum_\sigma\langle\ueta_\sigma, \dot {\underline m}_\sigma(s)\rangle\widehat {\widetilde 
F}_s(\ueta_-,\uxi,\ueta_+)\biggr| 
 \leq \frac 1N \mathcal D_N(H_{s}[\hat{f}_s],\widehat G_{s}^2,\uxi)\, .
\end{aligned}
\]

A straightforward computation (and the fact that $\vec m_{\sigma}(s)$ solves \eqref{eq:IVPM1}) verifies that 
\begin{align}
H_{s}[\hat{f}_s](\underline \xi, 0) &= 0\\
\nabla_{\vec\eta} H_{s}[\hat{f}_s]( 0,  0) &= \vec 0.
\end{align}
Moreover, we can define 
\begin{equation*}
	\|H\|_{3,1}:=\sup_{\|\vec 
\alpha_\eta\|_1\leq3,\|\vec\alpha_\xi\|\leq 
1}\bigl\|\partial_{\vec\eta}^{\vec\alpha_\eta}\partial_{\vec\xi}^{\vec\alpha_\xi}
H\bigr\|_\infty \, .
\end{equation*}
In \cite{BLP25}, we estimated $\|H_s(\hat f_s)\|_{3,1}$ uniformly in time in the special case where $p_+ = p_S = p_- =
0$. In Appendix \ref{app:moments}, we will present an improved and streamlined analysis of $\|H_s(\hat f_s)\|_{3,1}$.
Indeed, since the derivative of $\hat f_s$ are linked to the moments of $f_s$, by analyzing the latter we will show
that, for a suitable constant $C_2$,
\begin{equation*}
\sup_s \|H_s[\hat{ f}_s]\|_{3,1} \leq C_2M E_4(f_0)
\end{equation*}
where
\begin{equation*}
E_K(f_0):=\sup_{\underline\alpha\in\N^{3M}, \|\underline\alpha\|_1\leq K}\int \uv^{\underline\alpha}f(\uv)f\uv
\end{equation*}
and $\underline \alpha=(\vec\alpha_1,\ldots,\vec\alpha_M)$, while
\[
\uv^{\underline\alpha}=\prod_{i=1}^M \vec v_i^{\vec \alpha_i}\qquad\hbox{with}\qquad\vv^{\vec\alpha}=\prod_{i=1}^3 
v_i^{\alpha_i}
\]
and
\[
\|\underline \alpha\|_1=\sum_{i=1}^3\|\vec\alpha_i\|_1\qquad\hbox{with}\qquad\|\vec \alpha\|_1=\sum_{i=1}^3\alpha_i\,.
\]

Finally, a simple extension of our results in \cite{BLP25} gives rise to the following key technical Lemma: 

\begin{mylemma} \label{lem:prop5new} Let $H(\uxi, \veta) = H: \mathds \mathds R^p \times R^6 \to\mathds C$ be 
such that 
\[
\nabla_{\veta}H(0,0)=0
\] 
and, for every $\uxi\in\R^p$, 
\[
H(\uxi, 0)=0\, .
\] 
Assume moreover that 
\begin{equation*}
	\|H\|_{3,1} < \infty.	
\end{equation*}
Let also $g:\mathds R^6\mapsto\mathds C$ be such that, for some $T$, 
\[
|g(\vec\eta)|\leq \frac1{1+T\|\vec\eta^2\|}.
\]
Then for $\mathcal D_N(H,g,\uxi)$ defined in \eqref{eq:DN} we have 
\begin{equation*}
\mathcal D_N(H,g,\uxi) \leq C_1(T)\sqrt{N}
\|H\|_{3,1}^{\frac56} \mathcal D_1(H,g,\uxi)^{\frac16} .
\end{equation*}
for a suitable constant $C_1(T)$.
\end{mylemma}
 
\begin{Remark}
\emph{In \cite[Lemma 3.3]{BLP25} we proved the above statement for $H(\vec\xi, \vec \eta) = H: \mathds \mathds R^p
\times R^3 \to\mathds C$. It is not hard to check that the proof work essentially without modifications for the
situation of Lemma \ref{lem:prop5new}. }
\end{Remark}

We thus obtain
\begin{equation}\label{eq:final}
d_2(F_t, \widetilde F_t) \leq 
\frac 1 {\sqrt N }\int_0^t C\|H_s[\hat{ f}_s]\|_{3,1}^{\frac 56}\sup_{\underline \xi \ne 0} 
\left(
\mathcal D_1(H_s[\hat{ f}_s], \widehat G_{s}^2, \underline \xi)\right)^{\frac 16}ds
\end{equation}
It remains for us to understand $\mathcal D_1(H_s[\hat{f}_s], \widehat G_{\pm,s}^2, \underline \xi)$ and 
$\|H_s[\hat{f}_s]\|_{3,1}$. We will analyze the former in 
the next Subsection and, as we discussed above, 
the latter in Appendix \ref{app:moments}.

\subsection{Step 2: estimating $\mathcal D_1(H_s[\hat{f}_s], \widehat G_{\pm,s}^2, \underline \xi)$}\label{subsec:1part}
 
Observing that $\mathcal D_1(H_s[\hat{f}_s], \widehat G_{s}^2, \underline \xi)=\sup_{\veta_{1}} |H_s[
f_s](\uxi,\veta_{1})|/(\|\uxi\|^2+\|\veta_{1}\|^2)$, it remains to understand
\[
D_1:=\sup_{\uxi,\veta_{1}} \frac{|H_s[\hat {f}_s](\uxi,\veta_{1})|}{\|\uxi\|^2+\|\veta_{1}\|^2}.
\]

We reached a similar point in $\cite{BLP25}$. In that work, we were able to directly relate $D_1$ to the $d_2$-distance
between the initial state and the stead-state. That is not possible here, since the $d_2$-distance between the initial
state and the steady state is infinite when $\vec m_S \ne 0$, which is an important part of this paper. Moreover, such
an analysis from \cite{BLP25} would obfuscate the various timescales present in the current evolution: the
equilibration of the system, the evolution of the reservoir velocities, and the evolution of the reservoir
temperatures. We will instead break-up $H_s[\hat{f}_s]$ into pieces which capture these different timescales.

We will do this in three steps. Observe first that, from Proposition \ref{Prop:Contraction} below, it follows that
\[
\lim_{t\to\infty}  f_t(\uv) = \Gamma_{\tau_\infty}^M(\uv)
\]
where $\tau_\infty$ is defined in \eqref{eq:infty}.
We thus define
\[
\Gamma^M_{\infty,s}(\uv) = \Gamma_{\tau_\infty}^M(\uv - \underline m_S(s)).
\]
That is, $\Gamma^M_{\infty,s}$ is the thermostated equilibrium state re-centered around $\underline m_S(s)$. Also let
\[
h_s=\Phi_{s,0}\Gamma^M_{\infty,0}(\uv)\, .
\]
That is, $h_s$ starts at the same initial state as $\Gamma^M_{\infty,s}$ but evolves according to the thermostated
evolution.

Since $H_s[f]$ is linear we can write
\[
\begin{aligned}
D_1&\leq \sup_{\uxi,\veta_{1}} \frac{|H_s[\hat {f}_s-\hat 
h_s](\uxi,\veta_{1})|}{\|\uxi\|^2+\|\veta_{1}\|^2}+
\sup_{\uxi,\veta_{1}} \frac{|H_s[\hat h_s-\widehat 
\Gamma^M_{\infty,s}](\uxi,\veta_{1})|}{\|\uxi\|^2+\|\veta_{1}\|^2}+
\sup_{\uxi,\veta_{1}} \frac{|H_s[\widehat 
\Gamma^M_{\infty,s}](\uxi,\veta_{1})|}{\|\uxi\|^2+\|\veta_{1}\|^2}=\\
&:= (I) + (II) + (III).
\end{aligned}
\]
Loosely speaking, we will show that $(I)$ is bounded because $\Phi_{t,s}$ is a contraction; we will show that $(II)$ is 
bounded via a second Duhamel expansion; and we will show that $(III)$ is bounded due to the explicit form of $H_s$ and 
$\Gamma^M_{\infty,s}$. This estimates will be based on several propositions. 

Our estimate for $(I)$, $(II)$ and $(III)$ will be in term of the shears $\delta_\sigma(s):=\vec m_\sigma(s)-\vec 
m_S(s)$ and temperatures gradients $\epsilon_\sigma(s):=\tau_\sigma(s)-\tau_S(s)$. Our first proposition relates these 
quantities with their initial values.

\begin{Prop}\label{Prop:Delta}
If $\vec m_\sigma$ and $\vec m_S$ solve the IVP \eqref{eq:IVPM1} then we have
\begin{equation}\label{eq:Deltaps}
\|\vec \delta_+(s)\|^2+\|\vec \delta_-(s)\|^2\leq e^{-\frac {2Ms}{3N}}\Delta_p\, .
\end{equation}
Similarly if $\tau_\sigma$ and $\tau_S$ solve the IVP \eqref{eq:IVPM1} then we have
\begin{equation}\label{eq:DeltaTs}
|\epsilon_+(s)|+|\epsilon_-(s)|\leq e^{-\frac {2Ms}{3N}}\left(2\Delta_T+\frac 13 \Delta_p\right)\, .
\end{equation}
$\Delta_p$ and $\Delta_T$ are defined in \eqref{eq:defp}. Moreover we have
\begin{equation}\label{eq:ttinf}
|\tau_\sigma(s)-\tau_\infty|\leq \frac 12(|\epsilon_+(s)|+|\epsilon_-(s)| + \|\vec \delta_+(s)\|^2+\|\vec 
\delta_-(s)\|^2)\, .
\end{equation}
\end{Prop}
\begin{proof} The simple proof of these statements is reported in Appendix \ref{app:proofs}.\end{proof}

It will be useful for us to be able to relate $d_2(H_s[g],0)$ to $d_2(g,0)$ where appropriate, so we consider that 
first. 

\begin{Remark}
\emph{In this Subsection and in the Appendices we will use the letter $C$ to indicate a generic constant independent of
$M$ and $N$. It is not supposed to have a fixed value even when it appears in the same formula multiple times.
}
\end{Remark}

\begin{Prop}\label{Prop:H}
Let $g:\R^{3M}\mapsto \R$ be such that $\hat 
g(0)=0$ and $\nabla \hat g(0)=0$ then
\begin{equation*}
\sup_{\uxi,\veta} \frac{|H_s[\hat g](\uxi,\veta)|}{\|\uxi\|^2+\|\veta\|^2}\leq CM(1+\Delta_p)d_2(g,0)\, .
\end{equation*}
\end{Prop}

\begin{proof}Let $h:\R^3\mapsto \R$ with $\hat 
h(0)=0$ and $\nabla \hat h(0)=0$. Then we have
\[
\begin{aligned}
\left|R[\hat h\widehat G^1_{\sigma,s}](\vxi,\veta)-R[\hat h\widehat G^1_{\sigma,s}](\vxi,0)\widehat 
G^1_{\sigma,s}(\veta)\right|=&
\left|\int_{\mathds S^2} \hat h(\vxi^*)\left(\widehat G^1_{\sigma,s}(\veta^*)-
\widehat G^1_{\sigma,s}(0^*)\widehat G^1_{\sigma,s}(\veta)\right)d\omega\right|\leq \\
&2 \int_{\mathds S^2}|\hat h(\vxi^*)|\leq \frac 23 d_2(h,0)\|\vxi\|^2
\end{aligned}
\]
while, using \eqref{eq:IVPM1}, we get
\[
\begin{aligned}
|\widehat G^1_{\sigma,s}(\veta)\langle \veta,\vec m_S(s)-\vec m_\sigma(s)\rangle \hat h(\vxi)|
\leq &
\widehat \Gamma_{\tau_\sigma(s)}(\veta)\|\veta\|\|\vec m_S(s)-\vec m_\sigma(s)\||\hat h(\vxi)|\leq \\
&(C\tau_\sigma(s)+\|\vec m_S(s)-\vec m_\sigma(s)\|^2)d_2(h,0)\|\vxi\|^2\, .
\end{aligned}
\]
The thesis now follows.

\end{proof}

In order to control $(I)$ and $(II)$, it will also be necessary for us to see that $\Phi_{t,s}$ is a contraction in 
$d_2$, which is the content of the following proposition, which is itself based on the analogous argument for the 
thermostat defined in \cite{BLP25}.

\begin{Prop}\label{Prop:Contraction}
For any $s \leq t$ and for any $f, g: \R^{3M} \to \R$ such that $\hat f (0) = \hat g(0)$ and $\nabla \hat f(0) = \nabla 
\hat g(0)$, we have
\[
d_2(\Phi_{t,s} f, \Phi_{t,s} g) \leq e^{-\frac 13(t - s)} d_2(f,g).
\]
\end{Prop}
\begin{proof}
From \eqref{eq:oldPsi} together with Lemma A1 and equation (75) in \cite{BLP25} we get that for any $t>0$
\[
d_2((\Psi(t) + \Lambda\, \Id) f, (\Psi(t) + \Lambda\,  \Id) g) \leq \left(\Lambda -\frac13\right)d_2(f,g)
\]
where $\Lambda=M(2+\lambda_S)+2N\lambda_R$.
Moreover, it is easily verified that, for $\tau$ small and $t \geq s$,
\[
\Phi_{t+ \tau,s} f - \Phi_{t + \tau,s} g = \tau\left(\Psi(t) + \Lambda\,  \Id\right)  \Phi_{t,s}(f - g) + (1 - \Lambda
\tau + O(\tau^2)) \Phi_{t,s}(f - g) .
\]
Thus we have, for $\tau$ small enough and $t > s$
\begin{align*}
\frac{d}{dt} d_2( \Phi_{t,s}f, \Phi_{t,s}g) &\leq d_2(\left(\Psi(t) + \Lambda\,  \Id\right) \Phi_{t,s}f, \left(\Psi(t) 
+ 
\Lambda\,  \Id\right) \Phi_{t,s}g) - \Lambda d_2(\Phi_{t,s}f, \Phi_{t,s}g) \\
&\leq -\frac13 d_2(\Phi_{t,s}f, \Phi_{t,s}g).
\end{align*}
The conclusion now follows via Gr\"onwall's lemma.
\end{proof}

 Combining
Proposition \ref{Prop:Contraction}, Proposition \ref{Prop:H} and \eqref{eq:centered} immediately yields
\[
(I) \leq C M (1+\Delta_p)e^{-\frac s3} d_2(f_0, h_0)=C M (1+\Delta_p)e^{-\frac s3} d_2(\tilde f_0, 
\Gamma^M_{\tau_\infty})\, 
.
\]

Next, we consider $(II)$. We first apply Proposition \ref{Prop:H} to obtain
\begin{equation}\label{eq:IIH}
(II) \leq C M \left(1 + \Delta_p\right) d_2(h_s, \Gamma^M_{\infty,s}).
\end{equation}
We will now use the Duhamel formula to estimate $d_2(h_s, f_{\infty,s})$. This will require an understanding of 
$\frac{d}{ds} \Phi_{t,s}$. 

\begin{Prop}
For any $t \geq s$ we have
\begin{equation}
\label{eq:ds}
\frac{d}{ds} \Phi_{t,s} = -\Phi_{t,s} \Psi(s).
\end{equation}
\end{Prop}

\begin{proof}
From the group properties of $\Phi$ we get, for $\tau$ small,
\[
\Phi_{t,s+s'}-\Phi_{t,s}=\Phi_{t,s+s'}(\Id - \Phi_{s+s',s}).
\] 
Dividing by $s'$ and taking the limit for $s'\to 0$ obtain \eqref{eq:ds}.
\end{proof}

With \eqref{eq:ds} in hand, we can apply the Duhamel formula and appeal to Proposition \ref{Prop:Contraction} to obtain
\[
\begin{aligned}
d_2(h_s, \Gamma^M_{\infty,s}) &\leq \int_0^s d_2\left(\Phi_{s,s'} \Psi(s') \Gamma^M_{\infty,s'},
\Phi_{s,s'}\frac{d}{ds'} \Gamma^M_{\infty,s'}\right) \ ds'\\
&\leq C\int_0^s  e^{-\frac 13(s - s')}d_2\left(\Psi(s') \Gamma^M_{\infty,s'},
\frac{d}{ds'} \Gamma^M_{\infty,s'}\right)ds' \, .
\end{aligned}
\]
We now turn our attention to estimating this term.

\begin{Prop}\label{Prop:tedious}
For all $s' \geq 0$ we have 
\begin{equation*} 
d_2\left(\Psi(s') \Gamma^M_{\infty,s'},
\frac{d}{ds'} \Gamma^M_{\infty,s'}\right) \leq \frac 1 6  \sum_{\sigma\in\{+,-\}}\big(|\tau_\sigma(s') - \tau_\infty| + 
2\|\vec\delta_\sigma(s')\|^2\big)
\end{equation*} 
\end{Prop}

\begin{proof}
We begin by observing that
\begin{align*}
\widehat\Psi(s') \widehat\Gamma^M_{\infty,s'}(\uxi) - &\widehat{\frac{d}{ds'}\Gamma^M_{\infty,s'}}(\uxi) =\\
& \sum_{k=1}^M \Gamma_{\infty,s'}^{M - 1}(\uxi^k) 
\bigg(\int_{\mathds S^2}\left(e^{i\langle\vec \delta_+(s'),\vec 0^*_k\rangle}\widehat \Gamma_{\tau_+(s')}(\vec
0^*_k)+e^{i\langle\vec \delta_-(s'),\vec 0^*_k\rangle}\widehat \Gamma_{\tau_-(s')}(\vec
0^*_k)\right)\widehat \Gamma_\infty(\vxi_k^*) -
\\ 
&\quad 2\widehat \Gamma_\infty(\vxi_k) - i\langle \dot{\vec m}_S(s'), \vxi_k\rangle \Gamma_\infty(\vxi_k)\bigg) =: 
\sum_{k=1}^M \widehat\Gamma_{\infty,s'}^{M - 1}(\uxi^k) D_2(\vxi_k)
\end{align*}
where $0^*_k=\langle\omega,\vxi_k\rangle\omega$.
It will suffices to understand $D_2(\vxi)$.
We begin by writing $\widehat \Gamma_{\tau_\sigma(s')}(0^*) = \widehat \Gamma_{\tau_\sigma(s')}(0^*) - \widehat 
\Gamma_{\infty}(0^*) + \widehat \Gamma_{\infty}(0^*)$ to obtain
\begin{align*}
|D_2(\vxi)|&= \bigg|\int_{\mathds S^2} e^{i\langle\vec \delta_+(s'),\vec 0^*\rangle}( \widehat 
\Gamma_{\tau_+(s')}(\vec 0^*)  - \widehat \Gamma_{\infty}(0^*))
\widehat\Gamma_{\infty}(\vxi^*) \, d\omega 
 + 
\int_{\mathds S^2} e^{i\langle\vec \delta_+(s'),\vec 0^*\rangle}\widehat \Gamma_{\infty}(\vxi) \, d\omega
\\
&\quad +
\int_{\mathds S^2} e^{i\langle\vec \delta_-(s'),\vec 0^*\rangle}( \widehat \Gamma_{\tau_-(s')}(\vec 0^*)  - \widehat 
\Gamma_{\infty}(0^*))
\widehat\Gamma_{\infty}(\vxi^*) \, d\omega 
 + 
\int_{\mathds S^2} e^{i\langle\vec \delta_-(s'),\vec 0^*\rangle}\widehat \Gamma_{\infty}(\vxi) \, d\omega
\\
&\quad -
2\widehat \Gamma_\infty(\vxi) - i\langle \dot{m}_S(s'), \vxi\rangle \widehat\Gamma_\infty(\vxi) \bigg| 
\\
&\leq 
\int_{\mathds S^2} \left| \widehat \Gamma_{\tau_+(s')}(\vec 0^*)  - \widehat \Gamma_{\infty}(0^*)\right|
 \, d\omega + \int_{\mathds S^2} \left| \widehat \Gamma_{\tau_-(s')}(\vec 0^*)  - \widehat \Gamma_{\infty}(0^*)\right|
 \, d\omega
 \\
 &\quad + 
 \left|\int_{\mathds S^2} e^{i\langle\vec \delta_+(s'),\vec 0^*\rangle} + e^{i\langle\vec \delta_-(s'),\vec 0^*\rangle} 
 \, d\omega - 2 - i\langle \dot{m}_S(s'), \vxi\rangle \right|
\end{align*}
Using that $\int_{\mathds S^2}\|0^*\|^2d\omega=\frac13\|\vxi\|^2$, the first term in this estimate satisfies
\[
 \int_{\mathds S^2} \left| \widehat \Gamma_{\tau_\sigma(s')}(\vec 0^*)  - \widehat \Gamma_{\infty}(0^*)\right|
 \, d\omega \leq \frac 1 3 d_2(\Gamma_{\tau_\sigma(s')}, \Gamma_\infty)\|\vxi\|^2 \leq \frac 1 6 |\tau_\sigma(s') - 
 \tau_\infty|\|\vxi\|^2 \, .
\]
We can control the second term by expanding $e^{i\scal{\vec \delta_\sigma(s'),\vec 0^*}}$ using 
\begin{equation*}
e^{x} = 1 + x + x^2\int_0^1 e^{tx}(1 - t) \ dt
\end{equation*}
and using \eqref{eq:IVPM1} to obtain
\begin{align*}
\bigg|\int_{\mathds
S^2} &e^{i\langle\vec \delta_+(s'),\vec 0^*\rangle} + e^{i\langle\vec \delta_-(s'),\vec 0^*\rangle}
-2 - i\langle \dot{m_S}, \xi\rangle \bigg|
\\
&=
\bigg|
\int_{\mathds S^2} \langle\vec \delta_+(s'),\vec 0^*\rangle^2 \int_0^1 e^{i\langle\vec \delta_+(s'),\vec 0^*\rangle t} 
(1 - 
t)\, dt \, d\omega + 
\int_{\mathds S^2} \langle\vec \delta_-(s'),\vec 0^*\rangle^2\int_0^1 e^{i\langle\vec \delta_-(s'),\vec 0^*\rangle t} 
(1 - t)\, 
dt \, d\omega \bigg|
\\
&\leq 
\int_{S^2} \langle\vec \delta_+(s'),\vec 0^*\rangle^2 + \langle\vec \delta_-(s'),\vec 0^*\rangle^2 \, d\omega 
\leq
\frac 1 3 \left(\|\vec \delta_+(s')\|^2  + \|\vec \delta_-(s')\|^2\right) \|\xi\|^2 
\end{align*}

Combining these estimates we finally get the thesis.
\end{proof}

Proposition \ref{Prop:tedious}, Proposition \ref{Prop:Delta} and \eqref{eq:IIH} now yield
\[
(II) \leq C M (1 + \Delta_p) (\Delta_p+\Delta_T)e^{-\frac {2Ms}{3N}}\, .
\]

It remains for us to understand $(III).$ 
\begin{Prop}\label{Prop:Theta}
For any $s \geq 0$ we have
\[
 (III)\leq
C M\sum_\sigma  \left(|\tau_\sigma(s) - \tau_\infty| + \|\vec \delta_\sigma(s')\|^2\right)
\]
\end{Prop}
\begin{proof} The proof is very similar to that for Proposition \ref{Prop:tedious}. We refer readers to
Appendix~\ref{app:proofs} for full details.\end{proof}

Combining the estimates for $(I)$, $(II)$, and $(III)$, we obtain
\begin{equation}\label{eq:slow}
D_1\leq CM(1+\Delta_p) \left(e^{-\frac s3}d_2(\tilde f_0, h_0) +(\Delta_T+\Delta_p)e^{-\frac {2Ms}{3N}}\right)
\end{equation}
and thus, taking into account \eqref{eq:dottau} and Proposition \ref{Prop:Delta}, we finally get
\begin{equation}\label{eq:fullt2}
d_2(F_t, \overline F_t) \leq C\frac M{\sqrt N}(1+\Delta_p)E_4(f)^{\frac56}\left(d_2(\tilde f_0, 
\Gamma^M_{\tau_\infty})^{\frac 
16}\left(1-e^{-\frac t{18}}\right)+
\frac {9N}M\left(1-e^{-\frac {M}{9N}t}\right)(\Delta_p+\Delta_T)^{\frac 
16}\right)
\end{equation}
This concludes the proof of Theorem \ref{thm:2thermo}. 

To obtain Corollary \ref{coro:1thermo} we observe that, if $T_+=T_-=T$ and $\vec p_+=\vec p_-$ then 
$\tau_+(t)=\tau_-(t)=:\tau(t)$ and $\vec m_+(t)=\vec m_-(t)=:\vec m(t)$ for every $t\geq 0$. The IVP \eqref{eq:IVPM1} 
now reads
\[
\frac{d}{dt}\begin{pmatrix}
		 \vec m_S \\ \vec m
	\end{pmatrix}=
\frac 13\begin{pmatrix}
-2 & 2\\
\frac{M}{N} & -\frac MN
\end{pmatrix}
\begin{pmatrix}
	 \vec m_S \\ \vec m
	\end{pmatrix},
\qquad\qquad
\begin{pmatrix}
	 \vec m_S(0) \\ \vec m(0)
	\end{pmatrix}=
\begin{pmatrix}
	 \vec p_S \\ \vec p
	\end{pmatrix}
\]
from which we get
\[
\vec m(t)-\vec m_S(t)=e^{-\frac13\left(2+\frac MN\right)t}(\vec p-\vec p_S)\, .
\]
With a similar argument we get
\[
|\tau(s)-\tau_S(s)|\leq e^{-\frac13\left(2+\frac MN\right)t}\left(\Delta_T+\frac 13 \Delta_p\right)\, .
\]
Using again the estimates for $(I)$, $(II)$, and $(III)$ conclude the proof.

\section*{Declarations}

\subsection*{Conflict of Interest}
The authors have no conflicts of interest to disclose.

\subsection*{Author Contributions}
Federico Bonetto: Formal analysis (equal); Writing – original draft (equal); Writing – review \& editing (equal). \\
Matthew Powell: Formal analysis (equal); Writing – original draft (equal); Writing – review \& editing (equal).

\subsection*{Data Availability}
Data sharing is not applicable to this article as no new data were created or analyzed in this study.

\bibliographystyle{unsrt}

\appendix

\section{Estimating moments of $F_s$}\label{app:moments}

A key observation from \cite{BLP25} was that, in the standard static thermostat case, there is a constant $C$ 
independent of $M$ such that $\sup_s \|H_s[\hat{f}_s]\|_{3,1} \leq CM,$ which allowed us to obtain a uniform estimate 
in the one 
thermostat case and a `reasonable' short-time estimate in the two thermostat case. We will see, by a similar argument, 
that the same is true for the dynamic thermostats considered here.
 
\begin{Prop}\label{prop:H31}
 	There exists a constant $C$ such that $\sup_s \|H_s[\hat{f}_s]\|_{3,1} \leq CM E_4(f_0).$
\end{Prop}

Reasoning like in Lemmas B2 and B3 in \cite{BLP25} we can show that
\[
\|H_s[\hat{f}_s]\|_{3,1}\leq CM E_4(f_s)
\]
We thus have to study the evolution of the moments of
$f_s$. To do this, we will define an adjoint evolution on the space of polynomials in $\uv$. The properties of this 
evolution will show that $E_4(f_s)\leq C E_4(f_0)$ that proof Proposition \ref{prop:H31}.

Given $f,g:\R^3\to\R$, it is easy to see that
\[
\int g(\vv)B(t)[f](\vv)d\vv=\int B^\dagger(t)[g](\vv)f(\vv)d\vv
\]
where $B(t)$ is given by \eqref{eq:BM} while $B^\dagger(t)$ is the adjoint of $B(t)$ and is given by
\[
B^\dagger(t)[g](\vv)=\int_{\R^3}\mathcal M_t(\vu)\int_{\mathds S^2}g(\vv-(\vu-\vv,\omega)\omega)d\omega d\vu.
\]
Observe that, if $p(\vv)$ is a polynomial of degree $K$, than $p(\vv-(\vu-\vv,\omega)\omega)$ can be seen as a
polynomial in $\vv$ and $\vu$ with coefficients depending on $\omega$. It follows that $B^\dagger(t)[p](\vv)$ is a
polynomials in $\vv$ of degree $\leq K$ whose coefficients depend on the moments of $\mathcal M(t)$. Similarly we can
define the adjoint $R^\dagger$ of $R$ acting on function $f:\R^6\mapsto \R$. In this case it turn out that
$R^\dagger=R$. Moreover, for a similar reason, $R$ maps  the space of homogeneous polynomials of degree $K$ into itself.

We thus start by describing the action of $\cL_S$ and $\cL^\dagger_{B(t)}$ on polynomials of degree at most $K$. We
consider the space $W_o^K$ of homogeneous polynomials of degree $K$ in $\vv$ with the basis formed by the
$\vv^{\vec\alpha}$ with $\|\vec\alpha\|_1=K$ and the space $W^K=\bigoplus_{k=0}^K W_o^k$ of polynomials of degree $\leq 
K$
together with the natural projection $Q_k$ from $W^K$ to $W^k_o$ for $k\leq K$. Analogously we consider the space
$U_o^K$ of symmetric homogeneous polynomials of degree $K$ in $\vv_1$ and $\vv_2$ with the basis formed by the
$\vv_1^{\vec\alpha_1}\vv_2^{\vec\alpha_2}+\vv_1^{\vec\alpha_2}\vv_2^{\vec\alpha_1}$ with $\|\alpha\|_1+\|\alpha\|_2=K$
and the space $U^K=\bigoplus_{k=0}^K U_o^k$ of symmetric polynomials of degree $\leq K$.

We can now summarize the above discussion, together with a few observations from \cite{BLP25,BPP}:
\begin{enumerate}
\item For every $k$, $B^\dagger(t)W_o^k\subset W^k$, while $RU_o^k\subset U_o^k$. 

\item $B^k:=Q_k B^\dagger(t)\big|_{W^k_o}$ is independent of $\mathcal M(t)$ and thus of $t$.

\item $R-I$ acting on $U_o^k$ is negative semi-definite while, for $k>0$, $B^k-I$ is negative
definite\footnote{\label{foot:self}We say that a matrix $A$ is negative definite if $A+A^T$ has all negative
eigenvalues. In our case the transpose are taken w.r.t. the basis discussed in the previous paragraph. Observe that from
\cite{BPP} it follows that $B^k-I$ and $R-I$ are diagonalizable with all negative eigenvalues.}.

\end{enumerate}

Observe that point 1 means that $Q_k B^\dagger(t)\big|_{W^l_o}=0$ if $k<l$, that is $B^\dagger(t)$ is block upper 
triangular w.r.t. the $W^k_0$. Similarly $R$ is block diagonal w.r.t the $U^k_0$.

We now return to polynomials in $\uv$. For every $\underline \alpha\in (\N^3)^M$ with $\sum_i \|\vec \alpha_i\|_1=K$ we 
can associate the degree $K$ monomial $\uv^{\underline \alpha}=\prod_i \vv_i^{\alpha_i}$. Let $\underline n_{\underline 
\alpha}:\N^3/\{0,0,0\}\mapsto \N$ be the {\it occupation numbers} of $\underline \alpha$, that is $\underline 
n_{\underline 
\alpha}(\vec\beta)=\#\{i\,|\, \vec \alpha_i=\vec \beta\}$. Then the polynomials
\[
e_{\underline n}(\uv)=\mathcal S(\uv^{\underline \alpha})\qquad\hbox{with}\qquad \underline n_{\underline 
\alpha}=\underline n
\]
and where $\mathcal S$ is the symmetrization operator:
\[
\mathcal S[p](\uv):=\frac 1{M!}\sum_\pi p(\vv_{\pi(1)},\ldots,\vv_{\pi(M)}).
\]
form a basis in the space $V^K_o$ of homogeneous polynomials of degree $K$ in $\uv$. Let moreover $V^K=\bigoplus_{k=0}^K
V_o^K$ the space of polynomials in $\uv$ of degree at most $K$ and $P_k$ the natural projection from $V^K$ to $V^k_o$
for $k\leq K$.

\begin{myclaim}\label{claim:1} Let $d_K={\rm dim}(V^K)$ and $L_{B(t)}\in \R^{d_K\times d_K}$ be the matrix 
representation of 
$\sum_{i=1}^M (B_i^\dagger(t)-\mathrm{Id})$ in the basis formed by the $e_{\underline n}$ with $\sum_{\vec 
\alpha}\underline n(\vec \alpha)=K$. 
Analogously let $L_R\in\R^{d_K\times d_K}$ be the matrix representation of  $(M-1)^{-1}\sum _{i<j}^M 
(R_{i,j}-\mathrm{Id})$. We have

\begin{enumerate}
\item $L_{B(t)}$ is block upper triangular while $L_R$ is block diagonal w.r.t. the $V^k_o$.
\item $L_{B(t)}$ does not depend on $M$ while $L_R=\overline L_R+ M^{-1} P$ where $\overline L_R$ does not depend on 
$M$ while $P$ is bounded uniformly in $M$.
\item $P^{k}L_{B(t)}|_{V^{k}_o}$ does not depend on $t$.
\item $P^{k}(L_R+L_{B(t)})|_{V^{k}_o}$, $k>0$, is negative definite (see footnote \ref{foot:self}).
\end{enumerate}

\end{myclaim}

\begin{proof}
Points 1, 3, and 4 follow directly from the observations above. Regarding point 3 we observe that
\[
\cL^\dagger_{B(t)}e_{\underline n}(\uv)=\mathcal S(\cL^\dagger_{B(t)}\uv^{\underline \alpha})=\mathcal S\left(\sum_{i=1}^M B_i^\dagger(t)[\vv_i^{\vec \alpha_i}] \prod_{k\not=i}\vv_k^{\vec \alpha_k}\right)
\]
for any $\underline \alpha$ such that $\underline n_{\underline \alpha}=\underline n$. Let $\#(\underline \alpha)$ be 
the number of $\vec \alpha_i\not =(0,0,0)$ in $\underline \alpha$. Since $\#(\underline \alpha)\leq K<M$, given 
$\underline n$ we can chose $\underline \alpha$ such that $\vec \alpha_i\not=0$ for $i\leq \#(\underline \alpha)$ while 
$\vec \alpha_i=0$ for $i>\#(\underline \alpha)$.We thus get
\[
\cL^\dagger_{B(t)}e_{\underline n}(\uv)=\mathcal S\left(\sum_{i=1}^{\#(\underline\alpha)} B_i^\dagger(t)[\vv_i^{\vec 
\alpha_i}] \prod_{k\not=i}\vv_k^{\vec \alpha_k}\right)
\]
that clearly does not depend on $M$.
Similarly we can write
\[
\cL_{S}e_{\underline n}(\uv)=\mathcal S\left(\frac 1{M-1}\sum_{i<j} R_{i,j}[\vv_i^{\vec \alpha_i}\vv_j^{\vec \alpha_j}] 
\prod_{k\not=i,j}\vv_k^{\vec \alpha_k}\right)=\mathcal S\left(\sum_{i\leq\#{\underline \alpha}} R_{i,M}[\vv_i^{\vec 
\alpha_i}] \prod_{k\not=i,k\leq\#{\underline \alpha} }\vv_k^{\vec \alpha_k}\right)+O\left(\frac1M\right).
\]
\end{proof}

We can now return to the evolution of the moments of $f_s$. 
We get
\[
\int \uv^{\underline\alpha}f_s(\uv)d\uv=\int \Phi_{s,0}^\dagger[e_{{\underline 
n}_{\underline\alpha}}](\vv)f_0(\uv)d\uv=\sum_{\underline n} m_{\underline n} A(t)_{{\underline n},{\underline 
n}_{\underline\alpha}}
\]
where 
\[
m_{\underline n}=\int e_{\underline n}(\uv)f_0(\uv)d\uv
\]
while $A(t)$ solves the equation
\[
\dot{A}(t) = A(t)(L_S + L_{B(t)}^\dagger); \qquad A(0) = \rm Id\, .
\]
We thus get
\[
E_K(f_s)\leq \sqrt{d_K}E_K(f_0)\left\|A(t)\right\|_2\, .
\]
Calling $L^{l,k}(t)=P_k(L_S + L_{B(t)}^\dagger)\big|_{V_o^l}$ we get 
$L^{k,l}=0$, for $k<l$ so that
\begin{equation}\label{eq:recur}
\begin{aligned}
P_KA(t)&=e^{L^{K,K}t}P_K\\
P_{K-1}A(t)&=e^{L^{K-1,K-1}t}P_K+\int_0^s e^{L^{K-1,K-1}(t-s)}L^{K-1,K}(t)P_KA(t)\\
&\vdots\\
P_{k}A(t)&=e^{L^{k,k}t}P_k+\sum_{l=k+1}^K\int_0^s e^{L^{k,k}(t-s)}L^{l,k}(t)P_kA(t)\, .
\end{aligned}
\end{equation}

The thesis now follow easily by solving recursively \eqref{eq:recur} and observing that, from point 2 and 4 of Claim 
\ref{claim:1} and the continuity of eigenvalues and eigenvectors, it follows that $\|e^{L^{k,k}t}\|_2\leq Ce^{-ct}$, 
for suitable constant $C,c>0$.

\section{Proofs from Subsection \ref{subsec:1part}}\label{app:proofs}

We collect here the proofs of some of the propositions from Subsection \ref{subsec:1part}.

\begin{proof}[Proof of Proposition \ref{Prop:Delta}]
From \eqref{eq:IVPM1} we see that $\vec \delta_\sigma$ satisfy 
\[
\frac{d}{dt}\begin{pmatrix}
\vec \delta_+\\\vec \delta_-
\end{pmatrix}=
\frac13\begin{pmatrix} 1+\frac MN & -1\\ -1 & 1+\frac MN
\end{pmatrix}
\begin{pmatrix}
\vec \delta_+\\\vec \delta_-
\end{pmatrix}
\]
so that we get
\[
\begin{aligned}
\vec \delta_+(s)+\vec \delta_-(s)=&e^{-\frac s3\left(2+\frac MN\right)}(\vec \delta_+(0)+\vec \delta_-(0))\\
\vec \delta_+(s)-\vec \delta_-(s)=&e^{-\frac s3\frac MN}(\vec \delta_+(0)-\vec \delta_-(0))\, .
\end{aligned}
\]
Equation \eqref{eq:Deltaps} now follows using that
\[
\|\vec \delta_+\|^2+\|\vec \delta_-\|^2=\frac 12\left (\|\vec \delta_++\vec \delta_+\|^2+\|\vec 
\delta_+-\vec \delta_-\|^2\right)\, .
\]
Similarly we get
\[
\frac{d}{dt}\begin{pmatrix}
\epsilon_+\\\epsilon_-
\end{pmatrix}=
\frac13\begin{pmatrix} 1+\frac MN & -1\\ -1 & 1+\frac MN
\end{pmatrix}
\begin{pmatrix}
\epsilon_+\\\epsilon_-
\end{pmatrix}
 - \frac 19\begin{pmatrix} 1-\frac MN & 1\\ 1 & 1-\frac MN
\end{pmatrix}
\begin{pmatrix}
\|\vec \delta_+\|^2\\\|\vec \delta_-\|^2
\end{pmatrix}
\]
so that
\[
\begin{aligned}
\epsilon_+(s)+\epsilon_-(s)=&e^{-\frac s3\left(2+\frac MN\right)}(\epsilon_+(0)+\epsilon_-(0))-\frac 19\left( 2- \frac 
MN\right)\int_0^s e^{-\frac 
{s-s'}3\left(2+\frac MN\right)}(\|\vec \delta_+(s')\|^2+\|\vec \delta_-(s')\|^2)ds'\\
\epsilon_+(s)-\epsilon_-(s)=&e^{-\frac s3\frac MN}(\epsilon_+(0)-\epsilon_-(0))-\frac M{9N}\int_0^s e^{-\frac 
{(s-s')M}{3N}}(\|\vec \delta_+(s')\|^2-\|\vec \delta_-(s')\|^2)ds'\, .
\end{aligned}
\]
Using that
\[
\frac 12(|\epsilon_++\epsilon_-|+|\epsilon_+-\epsilon_-|)\leq |\epsilon_+|+|\epsilon_-|\leq 
|\epsilon_++\epsilon_-|+|\epsilon_+-\epsilon_-|
\]
we easily get \eqref{eq:DeltaTs}.

Finally we observe that, for every $s$,
\[
N^*\tau_\infty=N\tau_+(s)+N\tau_-(s)+M\tau_S(s)+N\|\vec m_+(s)-\vec m_S(s)\|^2+N\|\vec m_-(s)-\vec m_S(s)\|^2
\]
so that
\[
N^*(\tau_+(s)-\tau_\infty)=(N+M)(\tau_+(s)-\tau_S(s))-N(\tau_-(s)-\tau_S(s))+N(\|\vec \delta_+(s)\|^2+\|\vec 
\delta_-(s)\|^2)
\]
from which \eqref{eq:ttinf} follows.
\end{proof}
\medskip

\begin{proof}[Proof of Proposition \ref{Prop:Theta}]
For convenience of notation, we begin by writing
\[
\Theta_{\sigma,t}(x,y) := e^{i\langle x, \vec m_s(t)\rangle} \cdot e^{i\langle y, \vec m_\sigma(t)\rangle}.
\]
By recalling the definition of $H_s[\widehat \Gamma^M_{\infty,s}](\uxi,\veta_{1})$ and proceeding as in Proposition 
\ref{Prop:tedious}, we can write
\[
H_s[\widehat \Gamma^M_{\infty,s}](\uxi,\veta)=\sum_{\sigma=\pm}\sum_{k=1}^M \widehat \Gamma_{\infty,s}^{M - 
1}(\uxi^k) G_{\sigma',s}(\veta_{\sigma'})D_{3,\sigma}(\vxi_k,\veta_{\sigma})
\]
where
\[
\begin{aligned}\label{eq:firstexpansion}
D_{3,\sigma}(\vxi,\veta)&= 
\int_{\mathds
\mathds
\mathds S^2} \Theta_{\sigma, s}(\vxi^*, \veta^*) \widehat \Gamma_{\infty}(\vxi^*) \widehat\Gamma_{\sigma, s}(\veta^*) 
- 
\Theta_{\sigma,s}(\vxi^*, 0^*) \widehat\Gamma_{\infty}(\vxi^*) \widehat\Gamma_{\sigma, s}(0^*) \Theta_{\sigma, s}(0, 
\veta) \widehat\Gamma_{\sigma, s}(\veta) \, d\omega 
\\
&\quad\quad- 
i\frac 1 3 \Theta_{\sigma, s}(\vxi, \veta) \widehat \Gamma_{\sigma, s}(\veta)\widehat\Gamma_{\infty}(\vxi)\langle 
\veta, \vec m_S(s) - \vec m_\sigma(s)\rangle
\end{aligned}
\]
By writing 
$
\widehat \Gamma_{\sigma, s}(\veta) = \widehat \Gamma_{\sigma, s}(\veta) - \widehat \Gamma_\infty(\veta) + \widehat 
\Gamma_\infty(\veta)
$ 
we have
\begin{align}
\begin{split}\label{eq:secondexpansion}
\Theta_{\sigma, s}(\vxi^*, \veta^*) \widehat \Gamma_{\infty}(\vxi^*) \widehat\Gamma_{\sigma, s}(\veta^*) 
&- 
\Theta_{\sigma,s}(\vxi^*, 0^*) \widehat\Gamma_{\infty}(\vxi^*) \widehat\Gamma_{\sigma, s}(0^*) \Theta_{\sigma, s}(0, 
\veta) \widehat\Gamma_{\sigma, s}(\veta)
\\
&= 
\Theta_{\sigma, s}(\vxi^*, \veta^*) \widehat \Gamma_{\infty}(\vxi^*) \left(\widehat \Gamma_{\sigma, s}(\veta^*) - 
\widehat \Gamma_\infty(\veta^*)\right)
\\
&\quad- 
\Theta_{\sigma,s}(\vxi^*, 0^*) \widehat\Gamma_{\infty}(\vxi^*) \left(\widehat \Gamma_{\sigma, s}(0^*) - \widehat 
\Gamma_\infty(0^*)\right) \Theta_{\sigma, s}(0, \veta) \widehat\Gamma_{\sigma, s}(\veta)
\\
&\quad+
\widehat \Gamma_{\infty}(\vxi)\left(\Theta_{\sigma, s}(\vxi^*, \veta^*) \widehat\Gamma_{\infty}(\veta)
- 
\Theta_{\sigma,s}(\vxi^*, 0^*) \Theta_{\sigma, s}(0, \veta) \widehat\Gamma_{\sigma, s}(\veta) \right) \, .
\end{split}
\end{align}
Now observe that
\begin{align}
\label{eq:firstestimate}
\left| \Theta_{\sigma, s}(\vxi^*, \veta^*) \widehat \Gamma_{\infty}(\vxi^*) \left(\widehat \Gamma_{\sigma, s}(\veta^*) 
- \widehat \Gamma_\infty(\veta^*)\right) \right| 
\leq  \|\veta^*\|^2 d_2(\Gamma_{\sigma, s}, \Gamma_{\infty}) 
= 
\frac 12  \|\veta^*\|^2 |\tau_{\sigma}(s) - \tau_\infty|
\end{align}
and similarly
\begin{align}
\begin{split}
\label{eq:secondestimate}
\left| \Theta_{\sigma,s}(\vxi^*, 0^*) \widehat\Gamma_{\infty}(\vxi^*) \left(\widehat \Gamma_{\sigma, s}(0^*) - \widehat 
\Gamma_\infty(0^*)\right) \widehat\Gamma_{\sigma, s}(\veta)\Theta_{\sigma, s}(0, \veta) \right| 
&=
\frac 12 \|0^*\|^2 |\tau_\sigma(s) - \tau_\infty| \, .
\end{split}
\end{align}
Meanwhile, writing 
\[
\Theta_{\sigma, s}(x,y) = 1 + i\langle x, \vec m_s(t)\rangle + i\langle y, \vec m_\sigma(t)\rangle - \left(\langle x, 
\vec m_s(t)\rangle + \langle y, \vec m_\sigma(t)\rangle\right)^2 \int_0^1 \Theta_{\sigma, s}(tx, ty) (1 - t) \, dt
\]
and setting $\vec x=\langle\omega, \vxi\rangle\omega$ and  $\vec y=\langle\omega, \veta-\vxi\rangle\omega$, yields
\begin{equation}\label{eq:thirdestimate}
\begin{aligned}
|D_{3,\sigma}(\vxi,\veta)|&= 
\bigg|
\int_{\mathds
\mathds
\mathds S^2} i \langle\omega, \eta\rangle\langle\omega , \vec m_S(s) - \vec m_\sigma(s)\rangle - \frac i 3 \langle 
\veta, \vec m_S(s) - \vec m_\sigma(s) \rangle
\\
&\quad +
\langle\omega, \vxi\rangle^2\langle \omega  , \vec m_S(s) - \vec m_\sigma(s) \rangle^2 \int_0^1 
\Theta_{\sigma,s}(t\vec x, -t\vec x) (1 - t) \, dt
\\
&\quad -
\langle\omega, \veta - \vxi\rangle^2\langle\omega, \vec m_S(s) - \vec m_\sigma(s)\rangle^2 \int_0^1 \Theta_{\sigma, 
s}((t\vec y, -t\vec y) (1 - t) \, dt \, d\omega \bigg|
\\
&\leq 
4(\|\vxi\|^2 + \|\veta\|^2) \|\vec m_S(s) - \vec m_\sigma(s)\|^2 
\end{aligned}
\end{equation}

Combining \eqref{eq:firstestimate}, \eqref{eq:secondestimate}, and \eqref{eq:thirdestimate} with 
\eqref{eq:firstexpansion} and \eqref{eq:secondexpansion} yields
\begin{align*}
\frac{\left|H_s[\hat f_{\infty,s}](\uxi,\veta_{\pm,1})\right|}
{\|\uxi\|^2 
+ 
\|\veta\|}  
&\leq \sum_{\sigma}\sum_{j = 1}^M \frac{4(\|\uxi_j\|^2 + \|\veta\|^2)(|\tau_\sigma(s) - \tau_\infty| + \|\vec m_S(s) - 
\vec m_\sigma(s)\|^2)}{\|\uxi\|^2 
+ 
\|\veta\|} 
\\
&\leq 
C M\sum_\sigma  \left(|\tau_\sigma(s) - \tau_\infty| + \|\vec m_S(s) - \vec m_\sigma(s)\|^2\right) \, .
\end{align*}

\end{proof}

\end{document}